\pdfoutput=1
\documentclass[12pt,reqno]{article}
\usepackage[markup=underlined]{changes}
\RequirePackage{etex}

\usepackage{mathpazo}

\usepackage{etex}
\usepackage{graphicx}
\usepackage{setspace}
\usepackage{amsmath}
\usepackage{amsfonts}
\usepackage{titlesec}
\usepackage{pictex}
\usepackage{comment}
\usepackage{amsthm}
\usepackage{graphics,epsfig,verbatim,bm,latexsym,url,amsbsy}
\usepackage{rotating}
\usepackage[authoryear,round]{natbib}
\usepackage{float}
\usepackage[position=bottom]{subfig}
\usepackage{mathrsfs}
\usepackage{multirow}
\usepackage{array}
\usepackage{bigints}
\usepackage{bbm}
\usepackage{geometry}
\usepackage{bbm}
\usepackage[ruled,vlined]{algorithm2e}
\usepackage{soul}
\usepackage{mathtools}
\usepackage{amssymb}
\usepackage{enumitem}
\usepackage{tikz}
\usetikzlibrary{arrows.meta,decorations.pathreplacing}

\DeclarePairedDelimiterX{\inp}[2]{\langle}{\rangle}{#1, #2}

\usepackage{hyperref}
\hypersetup{
	colorlinks=true,
	linkcolor=red!60!black,
	citecolor=blue!60!black,
	filecolor=magenta,      
	urlcolor=blue!60!black,
}
\usepackage[nameinlink,noabbrev,sort,capitalise]{cleveref}

\newcommand{\Ex}{\mathbb{E}}

\crefname{thm}{Theorem}{Theorems}

\theoremstyle{definition} 
\theoremstyle{definition} 
\theoremstyle{definition} 

\theoremstyle{definition} 
\theoremstyle{definition} \newtheorem{definition}{Definition}
\theoremstyle{definition} 
\theoremstyle{definition} 
\theoremstyle{definition} \newtheorem{lemma}{Lemma}
\theoremstyle{plain} \newtheorem{theorem}{Theorem}
\theoremstyle{definition} 
\theoremstyle{definition} 
\theoremstyle{definition} 
\theoremstyle{definition} 
\theoremstyle{definition}
\theoremstyle{definition} 
\theoremstyle{definition} 
\theoremstyle{definition} 
\theoremstyle{definition} 
\theoremstyle{definition} 
\theoremstyle{definition}

\theoremstyle{definition} 
\usepackage{mathtools}
\titlespacing*{\paragraph}{0pt}{1.25ex plus 1ex minus .2ex}{0.5em}

\titleformat{\section}
		{\bfseries\center \MakeUppercase }
         {\thesection}
        {0.5em}
        {}
        []

\titleformat{\subsection}[runin]
        {\normalfont\bfseries}
        {\thesubsection}
        {0.5em}
        {\addperiod}
        []
\newcommand{\addperiod}[1]{#1.}
\allowdisplaybreaks

\title{\textbf{\Large{\MakeUppercase{
Technology Speed Limits}\thanks{First version: 1st Feb 2026; this version: \today. Koh: MIT Department of Economics; \url{ajkoh@mit.edu}; Sanguanmoo: MIT Department of Economics; \url{sanguanm@mit.edu}. We are grateful to Drew Fudenberg, Stephen Morris, Phillip Strack, and seminar participants at Berkeley, Columbia, LSE, and Yale for helpful comments. 
}}}}

\author{\makebox[.3\linewidth]{\small \MakeUppercase{Andrew Koh}}\\ 
\and \makebox[.3\linewidth]{\small \MakeUppercase{Sivakorn Sanguanmoo}}}

\date{}

\begin{document}

\maketitle 

\vspace{-2em} 
\begin{abstract}
    We study optimal technology regulation when private learning occurs both through doing (scaling up the technology) and through waiting (as time passes). We show that an adaptive speed limit---a cap on the rate at which the technology can increase per unit time---delivers optimal worst-case guarantees over all learning processes and/or preferences, and is the only time-consistent mechanism that does so. 
\end{abstract}

We learn about new technologies in two distinct ways. The first is \emph{by doing}: the process of scaling (e.g., through deployment or development) generates information about its impacts. The second is \emph{by waiting}: fixing the scale of the technology, the passage of time might, on its own, help us understand its promises and dangers. These margins of learning pose different tradeoffs: learning-by-doing risks irreversibly scaling up a dangerous technology; learning-by-waiting delays potential gains. 

We develop a new model of learning along both margins. An agent (e.g., AI firm) chooses a \emph{technology path} that specifies the extent to which the technology is deployed or developed through time. This path is shaped by (i) learning---as the agent scales up the technology or as time passes, she progressively learns about the technology's promises and dangers and alters her path accordingly; and (ii) policy---the transfers she faces as she traverses different technology paths. We are interested in how policy should be designed when the agent steering the technology is misaligned: it might not fully internalize harms to wider society, or might stand to gain disproportionately from the technology.

Our main results develop optimality foundations for \emph{adaptive speed limits}: a cap on the degree to which the technology can be scaled up per-unit-time. The magnitude of this cap evolves with new information generated by scaling up the technology or with the passage of time. We show that such policies are robustly optimal---they offer optimal worst-case guarantees---against different ways in which the firm might learn and/or preferences they have (\cref{thrm:speed_robustness}). Furthermore, among the universe of robust policies, this adaptive speed limit is uniquely time-consistent (\cref{thrm:speed_timecon}): it is the only one that can be implemented without commitment. 

\textbf{Related literature.} Our model of learning is new, and draws on work from probability on multiparameter filtrations \citep{khoshnevisan2006multiparameter} to formalize distinct but interconnected modes of learning. Thus, our agent faces a stochastic control problem that is substantially more general than those within the multi-arm bandit framework. Our contribution is to show that simple and interpretable mechanisms (speed limits) offer optimal payoff guarantees in an otherwise complicated environment.

Our work relates more broadly to the literature on technology regulation. Work here has taken learning seriously as an important policy consideration, but has done so only in the context of learning by waiting \citep{acemoglu2024regulating} \emph{or} learning by doing \citep{guerreiro2023regulating,gans2025learning,koh2024robust}.\footnote{See also \citet{trammell2024existential,jones2024ai} who analyze the case in which the risks are known, but limited deployment of a dangerous technology might nonetheless be optimal.} Our model unifies both kinds of learning within a single framework and derives novel policy prescriptions.

\section{Model} \label{sec:speed_model}

\paragraph{Payoffs} Let $\mathcal{L}, \mathcal{T} \subseteq \mathbb{R}_+$ denote compact subsets of the reals with $0\in\mathcal{L}$. Time is discrete so $\mathcal{T} = 1,2,\ldots T$ where $T := \max \mathcal{T}$. We use $l$ for deterministic technology levels and $\ell$ for stopping levels or path choices; similarly, $t$ denotes a fixed date, while $\tau$ is reserved for stopping times. The technology state is binary, $\Theta = \{0,1\}$, with common prior in $\Delta(\Theta)$. The principal's flow payoff $v: \Theta \times \mathcal{L} \to \mathbb{R}$ is continuously differentiable in $l$ and satisfies $v(\theta,0)=0$ for each $\theta$, with $v(1,\cdot)$ strictly increasing and $v(0,\cdot)$ strictly decreasing. $v$ can be interpreted as Bernoulli utility, but also might be the expectation of some underlying stochastic process (e.g., consumption) with law shaped by $\theta$. 

The agent's flow payoff is $u(\theta,l)=g(v(\theta,l))$ for some increasing convex function $g:\mathbb R\to\mathbb R$ with $g(0)=0$. For instance, this might reflect the presence of either (i) negative externalities: in state $\theta = 0$ the principal (whose preferences represent that of wider society) is harmed more; or (ii) disproportionate winners: if the state is $\theta = 1$ the agent captures more of the gains. Let $\mathbb{U}$ denote the set of all such agent preferences. Since $g$ is convex, the right derivative $\partial_+g(0) := \lim_{h \to 0^+} \frac{g(h) - g(0)}{h}$ exists and is positive. For a belief $\mu \in \Delta(\Theta)$ and technology level $l$, write the principal's and agent's expected flow payoffs as
\begin{align*}
V(\mu,l) &:= \Ex_{\theta \sim \mu}[v(\theta,l)], \\
U(\mu,l) &:= \Ex_{\theta \sim \mu}[u(\theta,l)].
\end{align*}

\paragraph{Learning over fields} 
A filtration on the \emph{field} is a family $\mathcal{F}:=(\mathcal{F}_{l,t})_{l,t}$ on a common probability space such that $l\geq l'$ and $t\geq t'$ imply $\mathcal{F}_{l,t}\supseteq\mathcal{F}_{l',t'}$. We assume these fields are coordinate-wise right-continuous. The principal's information is represented by a right-continuous field $\mathcal{G}:=(\mathcal{G}_{l,t})_{l,t}$. The agent learns weakly more: $\mathcal{G}_{l,t}\subseteq \mathcal{F}_{l,t}$ for all $(l,t)\in\mathcal{L}\times\mathcal{T}$, and $\mathbb{F}$ denotes the set of all such agent filtrations. We use $\mu_{l,t} := \Ex[\theta \mid \mathcal{F}_{l,t}]$ to denote the agent's belief at $(l,t)$. The collections 
\[
(\Ex[\theta | \mathcal{F}_{l,t}])_{l,t} \quad \text{and} \quad 
(\Ex[\theta | \mathcal{G}_{l,t}])_{l,t}
\] 
are random fields indexed by the partially ordered set $\mathcal{L} \times \mathcal{T}$; see \cite{khoshnevisan2006multiparameter} for a textbook treatment. 
For each fixed period $t$ and each $\mathcal{F}_{(\cdot,t)}$-stopping level $\ell$, define the stopped sigma-field
\[
\mathcal{F}_{\ell,t}
:=
\Big\{A: A\cap\{\ell\leq l\}\in \mathcal{F}_{l,t}
\text{ for every }l\in\mathcal{L}
\Big\}
\]
and we define $\mathcal{G}_{\ell,t}$ analogously. 

Our model incorporates both learning-by-doing (as the technology level increases) and learning-by-waiting (as time passes). The former is risky because technology is irreversible; the latter is costly because potential gains from the technology are delayed. We emphasize that we do not impose the `commutation' condition common in the literature on multiparameter filtrations \citep{cairoli1975stochastic} which requires that, for all integrable random variables $X$,
\[
\mathbb{E}\big[\mathbb{E}[X \mid \mathcal{F}_{l, \infty}] \mid \mathcal{F}_{\infty, t}\big] = \mathbb{E}[X \mid \mathcal{F}_{l, t}]
\]
where $\mathcal{F}_{l, \infty} := \bigvee_{s\in\mathcal T} \mathcal{F}_{l, s}$ and $\mathcal{F}_{\infty, t} := \bigvee_{l'\in\mathcal L} \mathcal{F}_{l', t}$. Intuitively, commutation requires
that the two sources of learning: learning-by-doing (increasing $l$) and learning-by-waiting (increasing $t$) deliver independent information. That is, fixing $l$ and running time beyond $t$ does not help predict---beyond what is already known at $(l, t)$---what will be learned from fixing $t$ and running the technology level beyond
$l$. This is restrictive in the context of learning about risky technologies: new information might retroactively change the interpretation of signals received at lower technology levels which introduces dependence between
learning-by-doing and learning-by-waiting.

\begin{figure}[H]
\centering
    \includegraphics[width=0.4\linewidth]{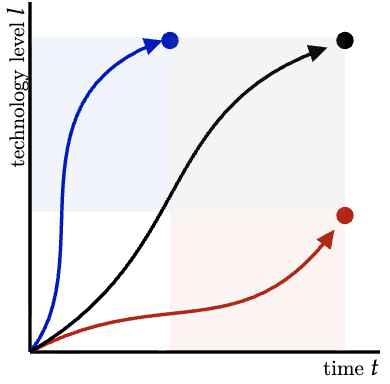}
    \caption{}
    \label{fig:learning_ordering}
\end{figure}

\cref{fig:learning_ordering} illustrates different technology paths over the level-time space. The filtration at the \textcolor{blue}{blue dot} cannot be ordered against that at the \textcolor{red}{red dot}, but both are weakly coarser than that at the black dot. Importantly, the join of the filtrations at the blue and red dots might still be coarser than that at the black dot e.g., when a side effect of a technology becomes visible only after enough physical time has passed.

\paragraph{Mechanisms} 
For a path $\pmb{\ell}=(\ell_t)_{t\in\mathcal T}$, write
$\pmb{\ell}_{\leq t}:=(\ell_1,\ldots,\ell_t)$ for the truncation of the technology path up to time $t$. A \emph{path-adapted mechanism} is a sequence $\phi=(\phi_t)_{t\in\mathcal T}$ of flow transfers that take values in $\mathbb R\cup\{+\infty\}$. At time $t$, after the current truncation $\pmb{\ell}_{\leq t}$ has been realized, the incremental transfer is $\phi_t(\pmb{\ell}_{\leq t})$. We require $\phi_t(\pmb{\ell}_{\leq t})$ to be $\mathcal G_{\ell_t,t}$-measurable. The agent's payoff net of the mechanism is thus
\[
\Ex\left[\sum_{t=1}^T \beta^{t-1}\left(U(\mu_{\ell_t,t},\ell_t)-\phi_t(\pmb{\ell}_{\leq t})\right)\right]
\]
where $\beta \in (0,1)$ is a discount factor. For a path $\pmb{\ell}$, write
\[
\Gamma^\phi(\pmb{\ell})
:=
\sum_{t=1}^T \beta^{t-1}\phi_t(\pmb{\ell}_{\leq t})
\]
for the total discounted transfer induced by $\phi$. Let $\Phi$ denote the set of path-adapted mechanisms $\phi$ such that, for each $t$, $\phi_t$ is pathwise lower semicontinuous in the realized truncation $\pmb{\ell}_{\leq t}$ on its finite-valued region, and the family
\[
\left\{\Gamma^\phi(\pmb{\ell}):
\pmb{\ell}\text{ is feasible and }\Gamma^\phi(\pmb{\ell})<+\infty
\right\}
\]
is uniformly integrable. These are standard optimal-stopping regularity conditions: lower semicontinuity of transfers makes the agent's net continuation payoff upper semicontinuous, while uniform integrability rules out limiting pathologies in expected payoffs. Hard constraints are represented by the value $+\infty$.

A special case is Markovian mechanisms in which the transfer $\phi_t(\pmb{\ell}_{\leq t})$ depends only on the current level-time pair $(\ell_t,t)$. We thus denote them by $\phi_{l,t}$. Markovian mechanisms condition only on the principal's knowledge $\mathcal{F}_{l,t}$ (adaptivity) but not the path she took to get there. \cref{fig:sample} illustrates the sample paths of four possible mechanisms: panel (a) illustrates a (potentially adaptive) cap on levels; panel (b) illustrates a (potentially adaptive) cap on levels that vary over time; panel (c) illustrates a linear Pigouvian tax over both levels and time; panel (d) illustrates a (potentially adaptive) nonlinear tax. 

\begin{figure}[H]  
\centering
    \caption{Possible mechanisms (sample paths)} 
    \subfloat[Cap on levels]{\includegraphics[width=0.35\linewidth]{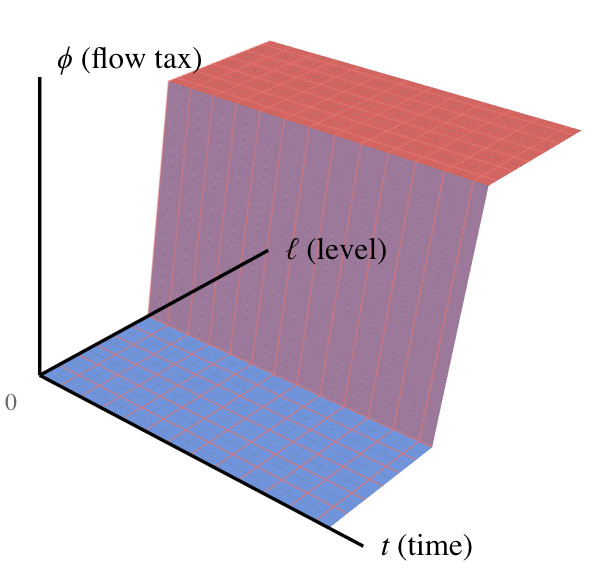}}
    \subfloat[Cap on levels \& time]{\includegraphics[width=0.35\linewidth]{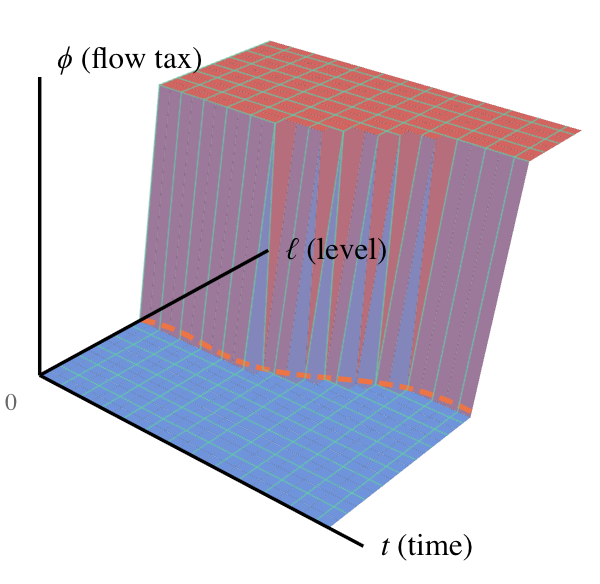}} \\ 
     \subfloat[Constant marginal tax]{\includegraphics[width=0.35\linewidth]{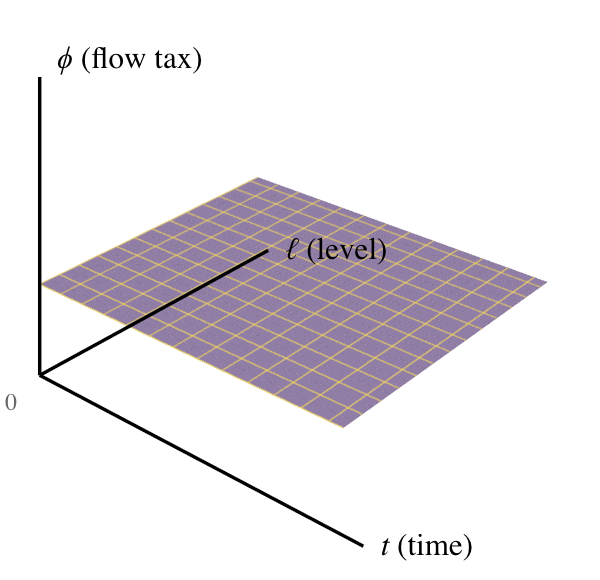}}  
    \subfloat[Increasing marginal tax]{\includegraphics[width=0.35\linewidth]{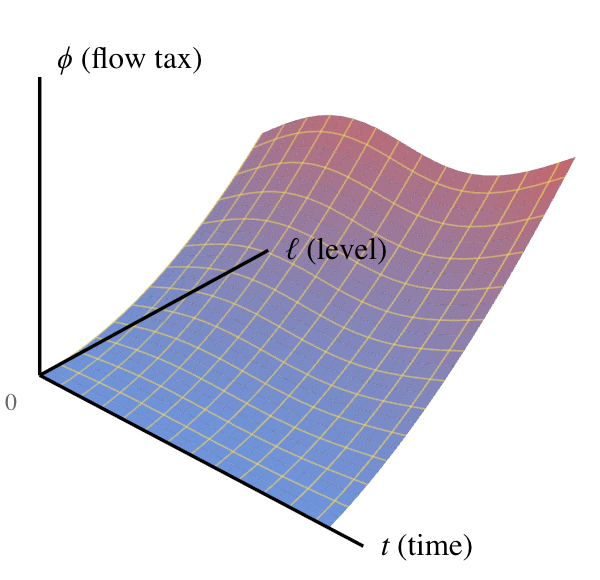}}
    \label{fig:sample}
\end{figure}

\paragraph{Agent's problem} 
A \emph{field-adapted path} (henceforth, just path) is a process $\pmb{\ell} := (\ell_t)_t$ such that (i) each $\ell_t$ is an $\mathcal{F}_{(\cdot,t)}$-stopping level; and (ii) $(\ell_t)_t$ is nondecreasing a.s. Paths are generalizations of stopping times for random fields, noting that our condition that $(\ell_t)_t$ is nondecreasing  reflects the irreversibility of technology. 

The agent optimizes over paths: 
\begin{align*}
\sup_{\pmb{\ell}}\, &\Ex \left[ \sum_{t=1}^T \beta^{t-1}\left(U(\mu_{\ell_t,t},\ell_t)-\phi_t(\pmb{\ell}_{\leq t})\right) \right] \\
&\text{s.t. $\pmb{\ell}$ is $\mathcal{F}$-adapted.} 
\end{align*}
We let $\pmb{\ell}^*(\phi, \mathcal{F},U)$ denote the largest optimal path for the agent. More explicitly, for each time $t \in \mathcal{T}$: (i) the agent chooses a stopping level $\ell^*_t  \geq \ell^*_{t-1}$ i.e., $\ell^*_t$ is a $\mathcal{F}_{(\cdot,t)}$ stopping level and past technology cannot be reversed; and (ii) $t < T$, time passes to $t + 1 \in \mathcal{T}$; otherwise payoffs are realized. 
\cref{fig:agent_problem} illustrates potential technology paths the agent might choose as she traverses the level-time space. 

\begin{figure}[H]
\centering
    \includegraphics[width=0.5\linewidth]{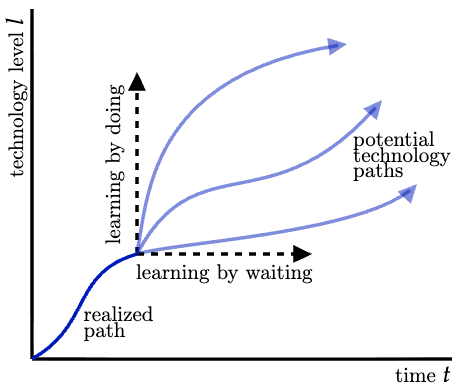}
    \caption{Possible agent technology paths}
    \label{fig:agent_problem}
\end{figure}

\paragraph{Robustness problems}
We are interested in the learning- and dually-robustness problems:
\[
\sup_{\phi \in \Phi} \inf_{\mathcal{F} \in \mathbb{F}} \Ex\left[\sum_{t=1}^T \beta^{t-1}V(\mu_{\ell^*_t,t},\ell^*_t)\right] \tag{L} \label{prob:learning}
\]
and 
\[
\sup_{\phi \in \Phi} \inf_{\substack{\mathcal{F} \in \mathbb{F} \\ U \in \mathbb{U}}} \Ex\left[\sum_{t=1}^T \beta^{t-1}V(\mu_{\ell^*_t,t},\ell^*_t)\right] \tag{D} \label{prob:dual}
\]
where $\pmb{\ell}^*=\pmb{\ell}^*(\phi,\mathcal F,U)$ is the agent's largest optimal path, with $U$ fixed in \eqref{prob:learning} and varied in \eqref{prob:dual}.
We call mechanisms that solve \eqref{prob:learning} learning-robust, and those that solve \eqref{prob:dual} dually-robust.

\section{Robust Speed Limits}\label{sec:speed_robustness}

\begin{definition}[Adaptive Speed Limit]\label{def:speed-limit-T}
Fix the path $\pmb{\bar \ell} := (\bar{\ell}_1, \ldots, \bar{\ell}_T)$ where: 
\begin{itemize}
    \item[(i)] \emph{Irreversibility:} for each $t < T$, $ \bar \ell_t \leq  \bar \ell_{t+1}$ a.s.; 
    \item[(ii)] \emph{Adaptivity:} each $\bar{\ell}_t$ is a $\mathcal{G}_{(\cdot,t)}$-stopping level.
\end{itemize}  
The adaptive speed limit $\pmb{\phi}$ induced by $\pmb{\bar \ell}$ is
\[
\phi_t(\pmb{\ell}_{\leq t}) = \begin{cases}
    0 \quad &\text{if $\ell_t \leq \bar \ell_t$}  \\
    +\infty &\text{otherwise.}
\end{cases}
\]
\end{definition}

Adaptive speed limits take the following form: at time $t$ following the past technology path $\pmb{\ell_{< t}} = (\ell_1,\ell_2, \ldots \ell_{t-1})$, the principal imposes a cap of $\bar \ell_t$ at time-$t$. Such mechanisms are adapted in the sense that the location of the time-$t$ cap depends on the `marginal' filtration $\mathcal{G}_{(\cdot,t)}$. Such mechanisms are speed limits in the sense that they impose a hard upper-bound $\bar \ell_t  - \ell_t \geq 0$ on how much the technology can scale between periods. \cref{fig:speed_limit} illustrates sample paths of speed limits (\textcolor{red}{red path}) and the agent's technology path (\textcolor{blue}{blue path}) given the speed limit.

\begin{figure}[H]  
\centering
    \caption{Illustration of speed limits (sample paths)} 
    \subfloat[]{\includegraphics[width=0.35\linewidth]{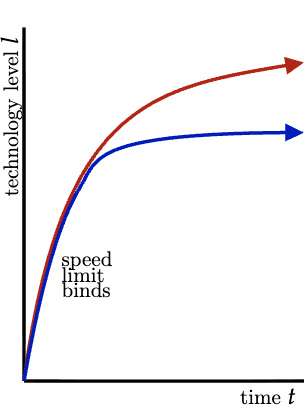}}
    \subfloat[]{\includegraphics[width=0.35\linewidth]{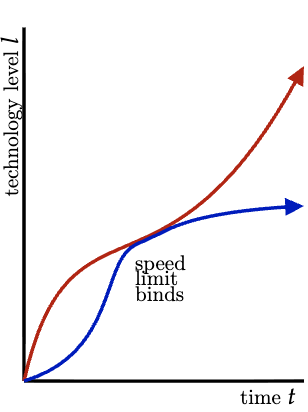}}
     \subfloat[]{\includegraphics[width=0.35\linewidth]{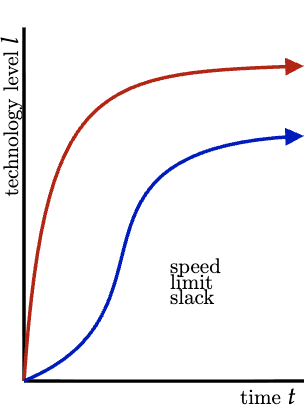}}  
    \label{fig:speed_limit}
\end{figure}

\begin{definition}[Principal's Direct Control Problem]
The principal's direct control problem is:
\[
\max_{(\ell_1,\ldots,\ell_T)} \mathbb{E}\left[\sum_{t=1}^{T} \beta^{t-1} V(\mu_{\ell_t,t}, \ell_t)\right]
\tag{Direct}
\label{eq:direct_control}
\]
subject to each $\ell_t$ being a $\mathcal{G}_{(\cdot,t)}$-stopping level with $\ell_t \geq \ell_{t-1}$ a.s.
Note that though the objective is written using the agent's belief $\mu_{\ell_t,t}$, any $\mathcal{G}$-adapted path has the same ex ante value when evaluated using the principal's belief by the linearity of $V$ and iterated expectations.
\end{definition}

Let $\pmb{\bar \ell} := (\bar{\ell}_1, \bar{\ell}_2, \ldots )$ denote the solution to \eqref{eq:direct_control}. Let $\bm{\phi}^*$ be the adaptive speed limit induced by $\bm{\bar \ell}$.

\begin{theorem}\label{thrm:speed_robustness}
The adaptive speed limit $\bm{\phi}^*$ is learning- and dually-robust.
\end{theorem}

The proof is involved and is deferred to \cref{appendix:proofs}. The basic idea is that an upper-bound on problems \eqref{prob:learning} and \eqref{prob:dual} can be constructed by setting $\mathcal{F} = \mathcal{G}$: since the agent has no private information, the principal cannot do better by steering the technology path herself, and this is exactly what solutions to the the principal's direct control problem \eqref{eq:direct_control} delivers. It remains to show that under any agent's preference $U \in \mathbb{U}$ and any agent's learning process $\mathcal{F} \in \mathbb{F}$, and for any path $\pmb{\ell}$ chosen by the agent facing $\pmb{\phi^*}$, the principal's payoff is weakly higher than this upper-bound. This is quite involved; we will briefly gesture at the broad intuition. 
When an agent goes \emph{slower} than the speed limit e.g., at time-$t$ following the path $\pmb{\ell}_{< t} := (\ell_1, \ldots \ell_{t-1})$ the agent chooses $\ell_t < \bar \ell_t$, she might do so for \emph{precautionary reasons}. 

\begin{figure}[H]  
\centering
    \caption{Option value of slowing down} 
    \subfloat[]{\includegraphics[width=0.45\linewidth]{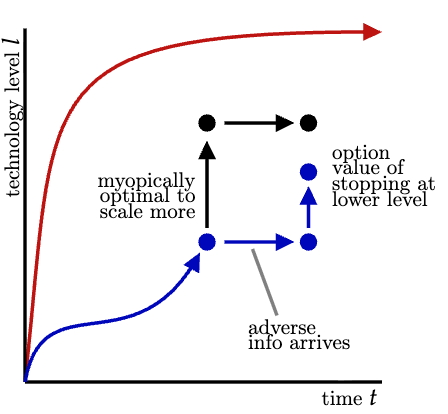}}
    \subfloat[]{\includegraphics[width=0.45\linewidth]{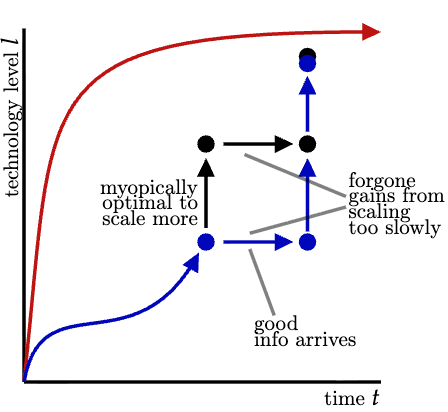}}
    \label{fig:option_comparison}
\end{figure}

That is, although she finds it myopically optimal to push up to the time-$t$ boundary (taking into account learning-by-doing as $l$ increases, but fixing time-$t$), choosing $\ell_t < \bar \ell_t$ delivers the option value of keeping the technology level low in the event that new and adverse information arrives at a future date $s > t$. This is illustrated in \cref{fig:option_comparison} panel (a) in which precaution (\textcolor{blue}{blue path}) in scaling up the technology at time $t$ can---on the depicted sample path---lead to an improvement over scaling up the technology more aggressively (\textbf{black path}), precisely because the technology is irreversible. On the other hand, the technology might turn out to be beneficial in which case precaution (\textcolor{blue}{blue path}) gives up gains in time-$t$ as depicted in panel (b). 

We develop new comparative statics to order these continuation values: if the agent prefers to scale the technology more slowly than the adaptive speed limit $\pmb{\phi^*}$, then path-by-path this delivers a weak improvement over the principal's direct control problem \eqref{eq:direct_control} which concludes the argument.



\section{Time-Consistency}\label{sec:speed_timecon}

We now turn to the question of whether the principal has incentives to follow through with her chosen policy. 

\begin{definition}[Continuation mechanisms]\label{def:continuation-mechanisms}
Let $\pmb{\ell}_{<t}:=(\ell_1,\ldots,\ell_{t-1})$ denote the agent's technology path (strictly) before time $t$, where we adopt the convention $\pmb{\ell}_{<1}:=\emptyset$. A \emph{decision history} is a pair $(\pmb{\ell}_{<t},\ell')$, where $\ell'$ is a stopped level at date $t$ satisfying $\ell'\geq \ell_{t-1}$ when $t>1$. A \emph{continuation path} from $(\pmb{\ell}_{<t},\ell')$ is $\pmb q=(q_s)_{s=t}^T$ such that $q_t\geq \ell'$ and $q_s\geq q_{s-1}$ for $s>t$. For $s\geq t$, write
\[
(\pmb{\ell}_{<t},\ell')\oplus \pmb q_{\leq s}
:=
(\ell_1,\ldots,\ell_{t-1},q_t,\ldots,q_s)
\]
for the induced truncation through date $s > t$. A \emph{continuation mechanism} after $(\pmb{\ell}_{<t},\ell')$ is a path-adapted mechanism applied to these concatenated truncations.
\end{definition}

\begin{definition}[Time-consistency]\label{def:timecon-paths}
A mechanism $\phi$ is \emph{time-consistent} if, after every decision history $(\pmb{\ell}_{<t},\ell')$, whether or not it is reached under $\phi$, its continuation mechanism is conditionally undominated, treating transfers before that history as sunk. That is, there is no continuation mechanism $\phi'$ after $(\pmb{\ell}_{<t},\ell')$ such that, for every $(\mathcal F,U)\in\mathbb F\times\mathbb U$, the largest induced continuation path $\pmb q'$ under $\phi'$ yields weakly higher conditional principal payoff than the largest optimal continuation $\pmb q^*$ under $\phi$, with strict inequality for some $(\mathcal F,U)$, conditional on $\mathcal G_{\ell',t}$:
\[
\underbrace{\mathbb E\!\left[
\sum_{s=t}^T \beta^{s-t}V(\mu_{q'_s,s},q'_s)
\;\middle|\;
\mathcal G_{\ell',t}
\right]}_{\substack{\text{Expected payoffs from $(\ell', t)$} \\ \text{from deviating to mechanism $\phi'$}}}
\geq
\underbrace{\mathbb E\!\left[
\sum_{s=t}^T \beta^{s-t}V(\mu_{q^*_s,s},q^*_s)
\;\middle|\;
\mathcal G_{\ell',t}
\right]}_{\substack{\text{Expected payoffs from $(\ell', t)$} \\ \text{from sticking with mechanism $\phi$}}}
\]
with strict inequality for some $(\mathcal F,U)$.
\end{definition}

Time-consistency requires mechanisms not to be dominated after any decision history. There are three distinct ways it can fail. The first resembles the logic of the Coase conjecture: the agent might, knowing that the regulator lacks commitment, deviate from continuing to stopping, or from stopping to continuing such as to influence the future mechanism. The second is because the set of learning processes and preferences the regulator regards as possible shrinks in the interim, so a mechanism that might have been previously rationalized for performing well at some learning-preference pair $(\mathcal{F},U) \in \mathbb{F} \times \mathbb{U}$ might become dominated as she now regards that as impossible. Finally, as the technology level increases, the regulator might simply lose incentives to follow through at interim histories. Time-consistency rules out all of these possibilities.

\begin{theorem}\label{thrm:speed_timecon}
    Suppose that the technology space is finite. Then the adaptive speed limit $\bm{\phi}^*$ is the unique time-consistent and dually-robust mechanism. 
\end{theorem}

The proof is fairly involved and is in \cref{appendix:proofs}. Showing uniqueness proceeds through a conditional characterization of dually-robust mechanisms. Such mechanisms must impose the same speed limit as $\pmb{\phi^*}$ and must give the agent the option to push up against that limit without being penalized (i.e., her interim transfers sum up to be weakly negative). Our key step is to show that if the agent's expected continuation transfers slope downward below the limit, modifying the continuation mechanism by `ironing' out the transfer keeps the mechanism dually-robust but dominates the original mechanism. This is illustrated in \cref{fig:ironing_transfers}.

\begin{figure}[H]
\centering
    \includegraphics[width=0.55\linewidth]{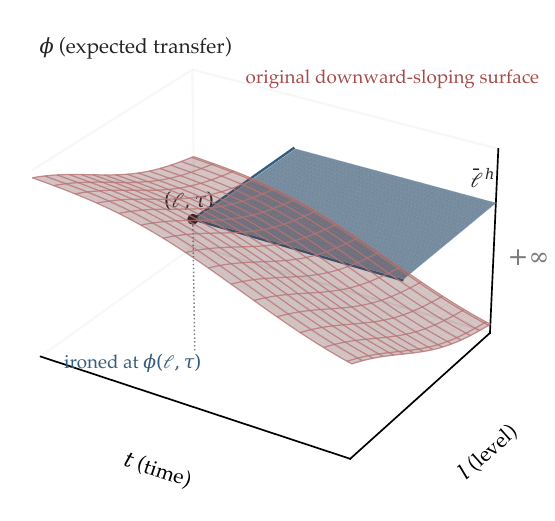}
    \caption{Ironing continuation transfers.}
    \label{fig:ironing_transfers}
\end{figure}

This ironing, in turn, generates incentives for the agent to stop prematurely (before the speed limit) which generates a strict improvement for the principal, thereby violating time-consistency. Then, proceeding via backward induction over the time-level grid we can show that all time-consistent mechanisms are equivalent to the adaptive speed limit $\pmb{\phi^*}$.

\section{Mitigating technological risks} \label{sec:mitigation} We have thus far focused on learning as a rationale for imposing speed limits. There is, of course, another reason speed matters---scaling up the technology slowly buys time for us to mitigate potential harms of the technology. This might be in the form of instituion building in which economic policies and legal rules take time to be implemented, or in the form technological solutions (e.g., AI safety research or vaccines) that might help to mitigate downside risks. Indeed, this view that we should `buy time' to develop mitigating measures is quite widespread among frontier AI labs.\footnote{See OpenAI: \url{https://openai.com/index/updating-our-preparedness-framework/}; DeepMind: \url{https://deepmind.google/blog/introducing-the-frontier-safety-framework/}; Anthropic: \url{https://www.anthropic.com/news/responsible-scaling-policy-v3}.}

We have deliberately abstracted from these considerations to focus on the role of learning \emph{qua} learning. But our model and results can incorporate the analysis of such mitigating measures without difficulty by allowing the technology state $\theta$ to transition from bad to good at some rate $\lambda(l,t)$ that is weakly decreasing in $l$ i.e., the harms from technologies at large scales are more difficult to mitigate. Now waiting becomes valuable for two reasons. As before, the passage of time teaches us about the risks (that is now changing through time). But it also buys time to mitigate the harms of the technology. Adaptive speed limits induced by the principal's solution to a modified direct control problem---that takes into account the endogenous transition of the technology state---is robustly optimal in such environments, and is the only time-consistent mechanism that is so.

\section{Discussion} \label{sec:speed_discussion}
We have developed a model in which learning about technology occurs both by doing (as $l$ increases) and by waiting (as $t$ increases). The former is risky because technology is irreversible; the latter is costly because we are impatient. The agent's technology path through time thus reflects an additional `extensive' margin that regulators might steer---restraint in the development or deployment of technology today delivers the option value of halting tomorrow if, in the intervening periods, we learn that the technology is too dangerous to develop further. We showed that adaptive speed limits are robust to the firm's learning process and/or preferences (\cref{thrm:speed_robustness}), and is the unique time-consistent robust mechanism (\cref{thrm:speed_timecon}). 

The idea that we should exercise caution in the face of an irreversible and dangerous technology is not new. Thus writes Samel Butler in 1872:
\begin{quote}
\emph{``I would repeat that I fear none of the existing machines; what I fear is the extraordinary rapidity with which they are becoming something very different to what they are at present. No class of beings have in any time past made so rapid a movement forward. Should not that movement be jealously watched, and checked while we can still check it?''} 
\end{quote}
Butler goes on to argue that it is necessary to ``destroy the more advanced of the machines''. His radicalism seems to be driven by pessimism that machines can be aligned to human values for ``the servant glides by imperceptible approaches into the master... [so that] man must suffer terribly on ceasing to benefit the machines.'' This is not our view---the point is that we simply do not know. Our goal has thus been to understand mechanisms that encourage learning without exposing society to excessive risk. 

More recently, in 2023 a number of prominent public figures signed a letter to pause `giant AI experiments', citing learning as a key rationale: 
\begin{quote}
\emph{``Powerful AI systems should be developed only once we are confident that their effects will be positive and their risks will be manageable. This confidence must be well justified and increase with the magnitude of a system’s potential effects.'' \citep{futureoflife2023pause}}
\end{quote}
From our present vantage point---and with the benefit of hindsight---this seems premature: more powerful AI systems have not caused significant harm, and continued scaling has taught us a great deal about their potential benefits and dangers. Perhaps we are still some distance from warranting any policy action, and perhaps policy action will not materialize even if warranted.\footnote{What seems politically feasible has also been shifting: \emph{``What is needed is a more active political involvement that is capable of slowing things down when everything is accelerating''} \citep{pope2026magnifica}.} Our point is simply that speed limits are a principled way to regulate the externalities that risk-seeking firms pose on wider society. But where these limits should be, who should descide, and how we might monitor and enforce them remain important questions for democratic deliberation.

\bibliography{speed_ref}

\appendix 

\crefalias{section}{appendix}
\section{Proofs} \label{appendix:proofs} 

\begin{proof}[Proof of \cref{thrm:speed_robustness}]
The proof is quite involved and uses a number of ideas from \citet{koh2024robust} as well as develops some new ideas. Throughout this proof we fix $(\mathcal{F},U)$ and suppress this dependence. In Steps 1--3, we analyze the adaptive speed limit $\bm{\phi}^*$; since $\bm{\phi}^*$ is Markovian and imposes zero transfers below the boundary, continuation values omit transfer terms. We start with some preliminaries. 

\underline{\textbf{\smash{Step 1: Preliminaries.}}}

The compactness of $\mathcal{L}$, continuity of flow payoffs, and lower semicontinuity and uniform integrability of mechanisms ensure that continuation payoff functions are upper semicontinuous on compact feasible sets. Standard optimal-stopping and measurable-maximum arguments then imply that the value functions and largest optimal selectors used in the continuation arguments below admit $\mathcal{F}_{\ell,t}$-measurable versions. 

\begin{lemma}[Stopped levels and pasting]\label{lem:stopped-levels}
Fix a period $t$. Stopping levels are closed under maxima. They are also closed under eventwise pasting: if $\ell$ and $\ell'$ are $\mathcal F_{(\cdot,t)}$-stopping levels and $A$ satisfies $A\cap\{\ell\leq l\},A\cap\{\ell'\leq l\}\in\mathcal F_{l,t}$ for every $l$, then $\ell 1_A+\ell'1_{A^c}$ is again an $\mathcal F_{(\cdot,t)}$-stopping level. Comparison events between stopped levels satisfy this measurability condition; in particular, events of the form $\{\ell>\ell'\}$ can be used for eventwise pasting.
\end{lemma}

\begin{proof}[Proof of \cref{lem:stopped-levels}]
For maxima,
\[
\{\ell\vee\ell'\leq l\}
=
\{\ell\leq l\}\cap\{\ell'\leq l\}
\in \mathcal F_{l,t}.
\]
For pasting,
\[
\{\ell 1_A+\ell'1_{A^c}\leq l\}
=
(A\cap\{\ell\leq l\})\cup(A^c\cap\{\ell'\leq l\})
\]
is in $\mathcal F_{l,t}$ by the stated measurability of the pasted event.

For the comparison-event claim, note that for each $l$,
\[
\{\ell>\ell'\}\cap\{\ell\leq l\}
=
\bigcup_{r\leq l}\big(\{\ell=r\}\cap\{\ell'<r\}\big),
\]
with the analogous expression for $\{\ell>\ell'\}\cap\{\ell'\leq l\}$. These events are in $\mathcal F_{l,t}$ by monotonicity of the filtration; for a general compact level space, the same argument is obtained by approximating from a countable dense grid and using right-continuity.
\end{proof}

\begin{lemma}[Stopped optional sampling]\label{lem:stopped-optional-sampling}
Fix a period $t$. If $\ell^-\leq \ell^+$ are $\mathcal F_{(\cdot,t)}$-stopping levels, then
\[
\mathbb E[\mu_{\ell^+,t}\mid \mathcal F_{\ell^-,t}]
=
\mu_{\ell^-,t}.
\]
\end{lemma}

\begin{proof}[Proof of \cref{lem:stopped-optional-sampling}]
For fixed $t$, $(\mu_{l,t})_{l\in\mathcal L}$ is a bounded martingale in the technology coordinate. The conclusion is the standard optional-sampling theorem for bounded martingales stopped at bounded stopping levels, applied to the stopped sigma-fields defined above.
\end{proof}

\begin{definition}[Continuation Values]\label{def:continuation-T}
For starting level $\ell$ at period $t$, define:
\begin{align}
J^A_t(\ell) &:= \sup \mathbb{E}\left[\sum_{s=t}^{T} \beta^{s-t} U(\mu_{\ell_s,s}, \ell_s) \,\Big|\, \mathcal{F}_{\ell,t}\right], \\
J^P_t(\ell) &:= \mathbb{E}\left[\sum_{s=t}^{T} \beta^{s-t} V(\mu_{\ell^*_s,s}, \ell^*_s) \,\Big|\, \mathcal{F}_{\ell,t}\right],
\end{align}
where the supremum is over feasible continuation paths $(\ell_t, \ldots, \ell_T)$ with $\ell_t \in [\ell, \bar{\ell}_t]$, $\ell_s \in [\ell_{s-1}, \bar{\ell}_s]$ for $s > t$, and each $\ell_s$ being $\mathcal{F}_{(\cdot,s)}$-adapted. The path $(\ell^*_t, \ldots, \ell^*_T)$ denotes the agent's optimal path starting from $\ell$ at period $t$.
\end{definition}

The following result comparing the magnitudes of preferences over stopping levels will also be helpful:

\begin{lemma}[Comparison of preference magnitudes]\label{lem:magnitude_compstat}
Fix a period $t$. Let $\ell^-$ be an $\mathcal{F}_{(\cdot,t)}$-stopping level and let $\ell^+ \geq \ell^-$ a.s.\ be another $\mathcal{F}_{(\cdot,t)}$-stopping level. Then:
\begin{align*}
V(\mu_{\ell^-,t}, \ell^-) &- \mathbb{E}[V(\mu_{\ell^+,t}, \ell^+) \mid \mathcal{F}_{\ell^-,t}] \\
&\geq \frac{1}{\partial_+g(0)} \Big( U(\mu_{\ell^-,t}, \ell^-) - \mathbb{E}[U(\mu_{\ell^+,t}, \ell^+) \mid \mathcal{F}_{\ell^-,t}] \Big).
\end{align*}
\end{lemma}

\begin{proof}[Proof of \cref{lem:magnitude_compstat}] The proof is fairly involved and follows the key steps of the corresponding one-dimensional argument in \citet{koh2024robust}, with an additional argument to handle the magnitudes of preferences over stopping levels. We proceed in four steps.

\medskip
\noindent\textbf{Step I: Decompose by state.}

Define the belief-weighted average continuation payoffs:
\[
\bar{v}_1 := \mathbb{E}\left[\frac{\mu_{\ell^+,t}}{\mu_{\ell^-,t}} \cdot v(1, \ell^+) \;\Big|\; \mathcal{F}_{\ell^-,t}\right], \qquad \bar{v}_0 := \mathbb{E}\left[\frac{1-\mu_{\ell^+,t}}{1-\mu_{\ell^-,t}} \cdot v(0, \ell^+) \;\Big|\; \mathcal{F}_{\ell^-,t}\right].
\]
By stopped optional sampling (\cref{lem:stopped-optional-sampling}), $\mathbb{E}[\mu_{\ell^+,t} \mid \mathcal{F}_{\ell^-,t}] = \mu_{\ell^-,t}$, so the weights $\frac{\mu_{\ell^+,t}}{\mu_{\ell^-,t}}$ and $\frac{1-\mu_{\ell^+,t}}{1-\mu_{\ell^-,t}}$ integrate to one and define probability measures.

The principal's and agent's payoff differences decompose as:
\begin{align*}
\Delta^P &:= V(\mu_{\ell^-,t}, \ell^-) - \mathbb{E}[V(\mu_{\ell^+,t}, \ell^+) \mid \mathcal{F}_{\ell^-,t}] \\
&= \mu_{\ell^-,t}(v(1, \ell^-) - \bar{v}_1) + (1-\mu_{\ell^-,t})(v(0, \ell^-) - \bar{v}_0),  \\
\Delta^A &:= U(\mu_{\ell^-,t}, \ell^-) - \mathbb{E}[U(\mu_{\ell^+,t}, \ell^+) \mid \mathcal{F}_{\ell^-,t}] \\
&= \mu_{\ell^-,t}(g(v(1, \ell^-)) - \bar{u}_1) + (1-\mu_{\ell^-,t})(g(v(0, \ell^-)) - \bar{u}_0), 
\end{align*}
where $\bar{u}_\theta := \mathbb{E}\left[\frac{\mathbb{P}(\theta \mid \mathcal{F}_{\ell^+,t})}{\mathbb{P}(\theta \mid \mathcal{F}_{\ell^-,t})} \cdot g(v(\theta, \ell^+)) \;\Big|\; \mathcal{F}_{\ell^-,t}\right]$ for $\theta \in \{0,1\}$.

\medskip
\noindent\textbf{Step II: Monotonicity and Jensen bounds.}

Since $\ell^+ \geq \ell^-$ a.s., $v(1, \cdot)$ is strictly increasing, and $v(0, \cdot)$ is strictly decreasing:
\begin{equation}\label{eq:vbounds}
\bar{v}_1 \geq v(1, \ell^-) > 0 \qquad \text{and} \qquad \bar{v}_0 \leq v(0, \ell^-) < 0.
\end{equation}
By Jensen's inequality applied to the convex function $g$ under the tilted measures:
\begin{equation}\label{eq:jensen}
\bar{u}_1 \geq g(\bar{v}_1) \qquad \text{and} \qquad \bar{u}_0 \geq g(\bar{v}_0).
\end{equation}

\medskip
\noindent\textbf{Step III: Term-by-term comparison via convexity.}

For notational convenience, let $v^-_\theta := v(\theta, \ell^-)$ and $v^+_\theta := \bar{v}_\theta$ for $\theta \in \{0,1\}$.

\textit{Good state ($\theta = 1$):} From \eqref{eq:vbounds}, $0 < v^-_1 \leq v^+_1$. Since $g$ is convex, the slope of any secant line is bounded below by any subgradient to the left. In particular, for $0 < a \leq b$:
\[
\frac{g(b) - g(a)}{b - a} \geq \partial_+g(a) \geq \partial_+g(0) = \partial_+g(0),
\]
where the second inequality uses the fact that the right derivative of a convex function is nondecreasing in its argument. Applying this with $a = v^-_1$ and $b = v^+_1$:
\[
g(v^+_1) - g(v^-_1) \geq \partial_+g(0)(v^+_1 - v^-_1),
\]
which can be rearranged to:
\begin{equation}\label{eq:state1_bound}
v^-_1 - v^+_1 \geq \frac{1}{\partial_+g(0)}(g(v^-_1) - g(v^+_1)).
\end{equation}
(Note both sides are non-positive.)

\textit{Bad state ($\theta = 0$):} From \eqref{eq:vbounds}, $v^+_0 \leq v^-_0 < 0$. Since $g$ is convex, for $a \leq b < 0$:
\[
\frac{g(b) - g(a)}{b - a} \leq \partial^-g(b) \leq \partial^-g(0) \leq \partial_+g(0) = \partial_+g(0),
\]
where $\partial^-g(b)$ denotes the left derivative at $b$, and we use the fact that the left derivative of a convex function is nondecreasing, and $\partial^-g(0) \leq \partial_+g(0)$ for any convex function. Applying this with $a = v^+_0$ and $b = v^-_0$:
\[
g(v^-_0) - g(v^+_0) \leq \partial_+g(0)(v^-_0 - v^+_0),
\]
which gives:
\begin{equation}\label{eq:state0_bound}
v^-_0 - v^+_0 \geq \frac{1}{\partial_+g(0)}(g(v^-_0) - g(v^+_0)),
\end{equation}
noting both sides are non-negative. 

\medskip
\noindent\textbf{Step IV: Combining the bounds.}

From \eqref{eq:jensen}, we have $-\bar{u}_\theta \leq -g(v^+_\theta)$, so:
\[
g(v^-_\theta) - \bar{u}_\theta \leq g(v^-_\theta) - g(v^+_\theta).
\]

Combining with \eqref{eq:state1_bound} and \eqref{eq:state0_bound}:
\[
v^-_\theta - v^+_\theta \geq \frac{1}{\partial_+g(0)}(g(v^-_\theta) - g(v^+_\theta)) \geq \frac{1}{\partial_+g(0)}(g(v^-_\theta) - \bar{u}_\theta)
\]
for each $\theta \in \{0, 1\}$.

Multiplying by the appropriate belief weights and summing:
\begin{align*}
\Delta^P &= \mu_{\ell^-,t}(v^-_1 - v^+_1) + (1-\mu_{\ell^-,t})(v^-_0 - v^+_0) \\
&\geq \frac{1}{\partial_+g(0)} \Big( \mu_{\ell^-,t}(g(v^-_1) - \bar{u}_1) + (1-\mu_{\ell^-,t})(g(v^-_0) - \bar{u}_0) \Big) \\
&= \frac{1}{\partial_+g(0)} \Delta^A. \qedhere
\end{align*}
\end{proof}

We are ready to prove \cref{thrm:speed_robustness}. 

\underline{\textbf{Step 2: Option Value Comparison}}

\begin{lemma}[No crossing of largest optimal continuations]\label{lemma:no-crossing}
Fix a period $t$ and let $\ell\leq \ell'\leq \bar{\ell}_t$ be $\mathcal{F}_{(\cdot,t)}$-stopping levels. Let $\ell^*_t(\ell)$ and $\ell^*_t(\ell')$ be the largest period-$t$ choices in the agent's optimal continuation paths under the adaptive speed limit. Then
\[
\ell^*_t(\ell)\vee \ell' \leq \ell^*_t(\ell')
\qquad\text{a.s.}
\]
\end{lemma}

\begin{proof}[Proof of \cref{lemma:no-crossing}]
Write the agent's period-$t$ payoff from choosing level $m$ and then continuing optimally as
\[
X_t(m):=
\begin{cases}
U(\mu_{m,T},m), & t=T,\\
U(\mu_{m,t},m)+\beta\Ex[J^A_{t+1}(m)\mid\mathcal F_{m,t}], & t<T.
\end{cases}
\]
Thus $\ell^*_t(x)$ is the largest solution to
\[
\sup_{m\geq x}\Ex[X_t(m)\mid \mathcal{F}_{x,t}],
\]
where $m$ ranges over feasible $\mathcal{F}_{(\cdot,t)}$-stopping levels satisfying the adaptive speed-limit constraint.

Let $a:=\ell^*_t(\ell)$, $b:=\ell^*_t(\ell')$, and $A:=\{a\vee \ell'>b\}$. Since $b\geq \ell'$ by feasibility, $A=\{a>b\}$. Suppose toward a contradiction that $A$ has positive probability. Define the pasted stopping levels
\[
d:=a\,1_A+b\,1_{A^c},
\qquad
c:=b\,1_A+a\,1_{A^c}.
\]
By \cref{lem:stopped-levels}, both $c$ and $d$ are feasible stopping levels; $d$ is feasible from $\ell'$ and $c$ is feasible from $\ell$. Moreover, $d\geq b$ and $d>b$ on $A$.

Because $b$ is optimal from $\ell'$,
\[
\Ex[X_t(b)\mid \mathcal{F}_{\ell',t}]
\geq
\Ex[X_t(d)\mid \mathcal{F}_{\ell',t}].
\]
Since $X_t(b)-X_t(d)=1_A(X_t(b)-X_t(a))$, this implies
\[
\Ex\!\left[1_A(X_t(b)-X_t(a))\mid \mathcal{F}_{\ell',t}\right]\geq 0.
\]
The inequality must be strict on a positive-probability event whenever $A$ has positive probability; otherwise $d$ would also be optimal from $\ell'$, contradicting that $b$ is the largest optimal choice from $\ell'$. Taking conditional expectations down to $\mathcal{F}_{\ell,t}$ gives
\[
\Ex[X_t(c)\mid \mathcal{F}_{\ell,t}]
-
\Ex[X_t(a)\mid \mathcal{F}_{\ell,t}]
=
\Ex\!\left[1_A(X_t(b)-X_t(a))\mid \mathcal{F}_{\ell,t}\right]
\geq 0,
\]
with strict inequality on a positive-probability event. This contradicts the optimality of $a$ from $\ell$. Therefore $A$ is null, and $\ell^*_t(\ell)\vee \ell'\leq \ell^*_t(\ell')$ a.s.
\end{proof}

\begin{lemma}[Option Value Comparison]\label{lem:option-value}
Fix $t \in \{1, \ldots, T\}$. Let $\ell\leq\ell'\leq\bar{\ell}_t$ be $\mathcal{F}_{(\cdot,t)}$-stopping levels. Then:
\[
J^P_t(\ell) - \mathbb{E}[J^P_t(\ell') \mid \mathcal{F}_{\ell,t}] 
\geq \frac{1}{\partial_+g(0)} \Big( J^A_t(\ell) - \mathbb{E}[J^A_t(\ell') \mid \mathcal{F}_{\ell,t}] \Big).
\]
On the event $\{\ell=\ell'\}$, both sides are zero.
\end{lemma}

\begin{proof}[Proof of \cref{lem:option-value}]
We induct backwards on time. It suffices to prove the result on $\{\ell<\ell'\}$, since the equality event is trivial and can be pasted back by \cref{lem:stopped-levels}.


\medskip
\noindent\textit{Base case ($t = T$):} At the final period, there is no continuation beyond $T$:
\begin{align*}
J^A_T(\ell) &= \sup_{\ell_T \in [\ell, \bar{\ell}_T]} \mathbb{E}[U(\mu_{\ell_T,T}, \ell_T) \mid \mathcal{F}_{\ell,T}], \\[6pt]
J^P_T(\ell) &= \mathbb{E}[V(\mu_{\ell^*_T(\ell),T}, \ell^*_T(\ell)) \mid \mathcal{F}_{\ell,T}],
\end{align*}
where $\ell^*_T(\ell)$ denotes the agent's optimal stopping level in period $T$ starting from $\ell$. This is exactly the baseline model of \citet{koh2024robust}.

\medskip
\noindent\textbf{Base step I: Construct bridging point.} Define $\tilde{\ell}_T := \ell^*_T(\ell) \vee \ell'$. Note that (i) $\tilde{\ell}_T \geq \ell^*_T(\ell)$ by construction; (ii) $\tilde{\ell}_T \geq \ell'$, so $\tilde{\ell}_T$ is feasible starting from $\ell'$; and (iii) $\tilde{\ell}_T \leq \bar{\ell}_T$ since both $\ell^*_T(\ell) \leq \bar{\ell}_T$ and $\ell' \leq \bar{\ell}_T$. 

\medskip
\noindent\textbf{Base step II: Compare $\ell^*_T(\ell)$ with $\tilde{\ell}_T$.}

Since $\ell^*_T(\ell) \leq \tilde{\ell}_T$ a.s., \cref{lem:magnitude_compstat} gives:
\begin{align*}
&V(\mu_{\ell^*_T(\ell),T}, \ell^*_T(\ell)) - \mathbb{E}[V(\mu_{\tilde{\ell}_T,T}, \tilde{\ell}_T) \mid \mathcal{F}_{\ell^*_T(\ell),T}] \\[6pt]
&\qquad \geq \frac{1}{\partial_+g(0)} \Big( U(\mu_{\ell^*_T(\ell),T}, \ell^*_T(\ell)) - \mathbb{E}[U(\mu_{\tilde{\ell}_T,T}, \tilde{\ell}_T) \mid \mathcal{F}_{\ell^*_T(\ell),T}] \Big).
\end{align*}

Taking $\mathbb{E}[\cdot \mid \mathcal{F}_{\ell,T}]$ and using the tower property:
\begin{equation}\label{eq:base-step2}
J^P_T(\ell) - \mathbb{E}[V(\mu_{\tilde{\ell}_T,T}, \tilde{\ell}_T) \mid \mathcal{F}_{\ell,T}]
\geq \frac{1}{\partial_+g(0)} \Big( J^A_T(\ell) - \mathbb{E}[U(\mu_{\tilde{\ell}_T,T}, \tilde{\ell}_T) \mid \mathcal{F}_{\ell,T}] \Big).
\end{equation}

\medskip
\noindent\textbf{Base step III: Compare $\tilde{\ell}_T$ with $\ell^*_T(\ell')$.}

By \cref{lemma:no-crossing}, $\tilde{\ell}_T \leq \ell^*_T(\ell')$ a.s. Hence \cref{lem:magnitude_compstat} gives:
\begin{align*}
&V(\mu_{\tilde{\ell}_T,T}, \tilde{\ell}_T) - \mathbb{E}[V(\mu_{\ell^*_T(\ell'),T}, \ell^*_T(\ell')) \mid \mathcal{F}_{\tilde{\ell}_T,T}] \\[6pt]
&\qquad \geq \frac{1}{\partial_+g(0)} \Big( U(\mu_{\tilde{\ell}_T,T}, \tilde{\ell}_T) - \mathbb{E}[U(\mu_{\ell^*_T(\ell'),T}, \ell^*_T(\ell')) \mid \mathcal{F}_{\tilde{\ell}_T,T}] \Big).
\end{align*}
Taking $\mathbb{E}[\cdot \mid \mathcal{F}_{\ell,T}]$:
\begin{equation}\label{eq:base-step3}
\begin{aligned}
&\mathbb{E}[V(\mu_{\tilde{\ell}_T,T}, \tilde{\ell}_T) \mid \mathcal{F}_{\ell,T}] - \mathbb{E}[J^P_T(\ell') \mid \mathcal{F}_{\ell,T}] \\[6pt]
&\qquad \geq \frac{1}{\partial_+g(0)} \Big( \mathbb{E}[U(\mu_{\tilde{\ell}_T,T}, \tilde{\ell}_T) \mid \mathcal{F}_{\ell,T}] - \mathbb{E}[J^A_T(\ell') \mid \mathcal{F}_{\ell,T}] \Big).
\end{aligned}
\end{equation}

\medskip
\noindent\textbf{Base step IV: Combine.} Adding \eqref{eq:base-step2} and \eqref{eq:base-step3}, the intermediate terms cancel:
\[
J^P_T(\ell) - \mathbb{E}[J^P_T(\ell') \mid \mathcal{F}_{\ell,T}] 
\geq \frac{1}{\partial_+g(0)} \Big( J^A_T(\ell) - \mathbb{E}[J^A_T(\ell') \mid \mathcal{F}_{\ell,T}] \Big).
\]

This completes the base case.


\medskip
\noindent\textit{Inductive step ($t+1 \to t$):} Assume the lemma holds for period $t+1$. We prove it for period $t$.

The continuation values decompose by the principle of optimality:
\begin{align*}
J^A_t(\ell) &= \mathbb{E}\left[U(\mu_{\ell^*_t(\ell),t}, \ell^*_t(\ell)) + \beta J^A_{t+1}(\ell^*_t(\ell)) \,\Big|\, \mathcal{F}_{\ell,t}\right], \\[6pt]
J^P_t(\ell) &= \mathbb{E}\left[V(\mu_{\ell^*_t(\ell),t}, \ell^*_t(\ell)) + \beta J^P_{t+1}(\ell^*_t(\ell)) \,\Big|\, \mathcal{F}_{\ell,t}\right],
\end{align*}
where $\ell^*_t(\ell)$ is the agent's optimal period-$t$ stopping level starting from $\ell$.

\medskip
\noindent\textbf{Inductive step I: Construct bridging point.} Define $\tilde{\ell}_t := \ell^*_t(\ell) \vee \ell'$. Note that (i) $\tilde{\ell}_t \geq \ell^*_t(\ell)$ by construction; (ii) $\tilde{\ell}_t \geq \ell'$, so $\tilde{\ell}_t$ is feasible starting from $\ell'$; and (iii) $\tilde{\ell}_t \leq \bar{\ell}_t$ since both $\ell^*_t(\ell) \leq \bar{\ell}_t$ and $\ell' \leq \bar{\ell}_t$. 

\medskip
\noindent\textbf{Inductive step II: Compare $\ell^*_t(\ell)$ with $\tilde{\ell}_t$.}

Since $\ell^*_t(\ell) \leq \tilde{\ell}_t$ a.s., \cref{lem:magnitude_compstat} gives:
\begin{align*}
&V(\mu_{\ell^*_t(\ell),t}, \ell^*_t(\ell)) - \mathbb{E}[V(\mu_{\tilde{\ell}_t,t}, \tilde{\ell}_t) \mid \mathcal{F}_{\ell^*_t(\ell),t}] \\[6pt]
&\qquad \geq \frac{1}{\partial_+g(0)} \Big( U(\mu_{\ell^*_t(\ell),t}, \ell^*_t(\ell)) - \mathbb{E}[U(\mu_{\tilde{\ell}_t,t}, \tilde{\ell}_t) \mid \mathcal{F}_{\ell^*_t(\ell),t}] \Big).
\end{align*}

By the inductive hypothesis at period $t+1$ with stopping levels $\ell^*_t(\ell)$ and $\tilde{\ell}_t$ (with the equality case being trivial):
\begin{align*}
&J^P_{t+1}(\ell^*_t(\ell)) - \mathbb{E}[J^P_{t+1}(\tilde{\ell}_t) \mid \mathcal{F}_{\ell^*_t(\ell),t+1}] \\[6pt]
&\qquad \geq \frac{1}{\partial_+g(0)} \Big( J^A_{t+1}(\ell^*_t(\ell)) - \mathbb{E}[J^A_{t+1}(\tilde{\ell}_t) \mid \mathcal{F}_{\ell^*_t(\ell),t+1}] \Big).
\end{align*}

Taking $\mathbb{E}[\cdot \mid \mathcal{F}_{\ell,t}]$ of both inequalities, applying the tower property, and combining with weight $\beta$ on the second:
\begin{equation}\label{eq:ind-step2}
\begin{aligned}
&J^P_t(\ell) - \mathbb{E}[V(\mu_{\tilde{\ell}_t,t}, \tilde{\ell}_t) + \beta J^P_{t+1}(\tilde{\ell}_t) \mid \mathcal{F}_{\ell,t}] \\[6pt]
&\qquad \geq \frac{1}{\partial_+g(0)} \Big( J^A_t(\ell) - \mathbb{E}[U(\mu_{\tilde{\ell}_t,t}, \tilde{\ell}_t) + \beta J^A_{t+1}(\tilde{\ell}_t) \mid \mathcal{F}_{\ell,t}] \Big).
\end{aligned}
\end{equation}

\medskip
\noindent\textbf{Inductive step III: Compare $\tilde{\ell}_t$ with $\ell^*_t(\ell')$.}

By \cref{lemma:no-crossing}, $\tilde{\ell}_t \leq \ell^*_t(\ell')$ a.s. Hence \cref{lem:magnitude_compstat} gives:
\begin{align*}
&V(\mu_{\tilde{\ell}_t,t}, \tilde{\ell}_t) - \mathbb{E}[V(\mu_{\ell^*_t(\ell'),t}, \ell^*_t(\ell')) \mid \mathcal{F}_{\tilde{\ell}_t,t}] \\[6pt]
&\qquad \geq \frac{1}{\partial_+g(0)} \Big( U(\mu_{\tilde{\ell}_t,t}, \tilde{\ell}_t) - \mathbb{E}[U(\mu_{\ell^*_t(\ell'),t}, \ell^*_t(\ell')) \mid \mathcal{F}_{\tilde{\ell}_t,t}] \Big).
\end{align*}

By the inductive hypothesis at period $t+1$:
\begin{align*}
&J^P_{t+1}(\tilde{\ell}_t) - \mathbb{E}[J^P_{t+1}(\ell^*_t(\ell')) \mid \mathcal{F}_{\tilde{\ell}_t,t+1}] \\[6pt]
&\qquad \geq \frac{1}{\partial_+g(0)} \Big( J^A_{t+1}(\tilde{\ell}_t) - \mathbb{E}[J^A_{t+1}(\ell^*_t(\ell')) \mid \mathcal{F}_{\tilde{\ell}_t,t+1}] \Big).
\end{align*}

Taking expectations conditional on $\mathcal{F}_{\ell,t}$ and combining:
\begin{equation}\label{eq:ind-step3}
\begin{aligned}
&\mathbb{E}[V(\mu_{\tilde{\ell}_t,t}, \tilde{\ell}_t) + \beta J^P_{t+1}(\tilde{\ell}_t) \mid \mathcal{F}_{\ell,t}] - \mathbb{E}[J^P_t(\ell') \mid \mathcal{F}_{\ell,t}] \\[6pt]
&\qquad \geq \frac{1}{\partial_+g(0)} \Big( \mathbb{E}[U(\mu_{\tilde{\ell}_t,t}, \tilde{\ell}_t) + \beta J^A_{t+1}(\tilde{\ell}_t) \mid \mathcal{F}_{\ell,t}] - \mathbb{E}[J^A_t(\ell') \mid \mathcal{F}_{\ell,t}] \Big).
\end{aligned}
\end{equation}

\medskip
\noindent\textbf{Inductive step IV: Combine.} Adding \eqref{eq:ind-step2} and \eqref{eq:ind-step3}, the intermediate terms cancel:
\[
J^P_t(\ell) - \mathbb{E}[J^P_t(\ell') \mid \mathcal{F}_{\ell,t}] 
\geq \frac{1}{\partial_+g(0)} \Big( J^A_t(\ell) - \mathbb{E}[J^A_t(\ell') \mid \mathcal{F}_{\ell,t}] \Big).
\]

Inducting on $t$ completes the proof.
\end{proof}

\underline{\textbf{Step 3: Combining periods}}

\begin{lemma}[Period $t$ Comparison]\label{lem:period-t}
For any period $t \in \{1, \ldots, T\}$ and any starting stopping level $\ell \in [0, \bar{\ell}_t]$:
\[
J^P_t(\ell) \geq \mathbb{E}\left[\sum_{s=t}^{T} \beta^{s-t} V(\mu_{\bar{\ell}_s, s}, \bar{\ell}_s) \,\Big|\, \mathcal{F}_{\ell,t}\right].
\]
\end{lemma}

\cref{lem:period-t} states that the principal's payoff under the agent's optimal play (starting from $\ell$ at period $t$; this is how $J^P_t(\ell)$ is defined) is at least as large as under direct control.

\begin{proof}[Proof of \cref{lem:period-t}]
By backward induction on $t$.

\medskip
\textbf{Base case ($t = T$):} We must show:
\[
J^P_T(\ell) \geq \mathbb{E}[V(\mu_{\bar{\ell}_T, T}, \bar{\ell}_T) \mid \mathcal{F}_{\ell,T}].
\]

Let $\ell^*_T(\ell)$ denote the agent's optimal choice starting from $\ell$, and let
\[
E_T:=\{\ell^*_T(\ell)=\bar{\ell}_T\}.
\]
On $E_T$, the agent's choice coincides with the direct-control choice. On $E_T^c$, agent optimality gives
\[
U(\mu_{\ell^*_T(\ell), T}, \ell^*_T(\ell)) \geq \mathbb{E}[U(\mu_{\bar{\ell}_T, T}, \bar{\ell}_T) \mid \mathcal{F}_{\ell^*_T(\ell), T}].
\]
Together with \cref{lem:magnitude_compstat}, this implies
\[
1_{E_T^c}\left(
V(\mu_{\ell^*_T(\ell), T}, \ell^*_T(\ell))
-\mathbb{E}[V(\mu_{\bar{\ell}_T, T}, \bar{\ell}_T) \mid \mathcal{F}_{\ell^*_T(\ell), T}]
\right)\geq 0.
\]
Hence, taking $\mathbb{E}[\cdot \mid \mathcal{F}_{\ell,T}]$ and applying the tower property,
\begin{align*}
J^P_T(\ell)
&= \mathbb{E}\!\left[1_{E_T}V(\mu_{\bar{\ell}_T, T}, \bar{\ell}_T)
+1_{E_T^c}V(\mu_{\ell^*_T(\ell), T}, \ell^*_T(\ell)) \,\Big|\, \mathcal{F}_{\ell,T}\right] \\
&\geq \mathbb{E}\!\left[1_{E_T}V(\mu_{\bar{\ell}_T, T}, \bar{\ell}_T)
+1_{E_T^c}\mathbb{E}[V(\mu_{\bar{\ell}_T, T}, \bar{\ell}_T) \mid \mathcal{F}_{\ell^*_T(\ell), T}] \,\Big|\, \mathcal{F}_{\ell,T}\right] \\
&= \mathbb{E}[V(\mu_{\bar{\ell}_T, T}, \bar{\ell}_T) \mid \mathcal{F}_{\ell,T}].
\end{align*}

\medskip
\textbf{Inductive step ($t+1 \to t$):} Assume the lemma holds for period $t+1$ i.e.,
\[
J^P_{t+1}(\ell') \geq \mathbb{E}\left[\sum_{s=t+1}^{T} \beta^{s-(t+1)} V(\mu_{\bar{\ell}_s, s}, \bar{\ell}_s) \,\Big|\, \mathcal{F}_{\ell',t+1}\right] \quad \text{for all } \ell' \in [0, \bar{\ell}_{t+1}].
\]

We will show:
\[
J^P_t(\ell) \geq \mathbb{E}\left[\sum_{s=t}^{T} \beta^{s-t} V(\mu_{\bar{\ell}_s, s}, \bar{\ell}_s) \,\Big|\, \mathcal{F}_{\ell,t}\right].
\]

Let $\ell^*_t(\ell)$ denote the agent's optimal period-$t$ choice starting from $\ell$, and let
\[
E_t:=\{\ell^*_t(\ell)=\bar{\ell}_t\}.
\]
On $E_t$, the agent's choice coincides with the direct-control choice in period $t$. On $E_t^c$, agent optimality gives:
\[
U(\mu_{\ell^*_t(\ell), t}, \ell^*_t(\ell)) + \beta J^A_{t+1}(\ell^*_t(\ell)) 
\geq \mathbb{E}[U(\mu_{\bar{\ell}_t, t}, \bar{\ell}_t) + \beta J^A_{t+1}(\bar{\ell}_t) \mid \mathcal{F}_{\ell^*_t(\ell), t}].
\]

By \cref{lem:magnitude_compstat} (applied to period-$t$ flow payoffs):
\begin{align*}
&V(\mu_{\ell^*_t(\ell), t}, \ell^*_t(\ell)) - \mathbb{E}[V(\mu_{\bar{\ell}_t, t}, \bar{\ell}_t) \mid \mathcal{F}_{\ell^*_t(\ell), t}] \\[6pt]
&\qquad \geq \frac{1}{\partial_+g(0)} \Big( U(\mu_{\ell^*_t(\ell), t}, \ell^*_t(\ell)) - \mathbb{E}[U(\mu_{\bar{\ell}_t, t}, \bar{\ell}_t) \mid \mathcal{F}_{\ell^*_t(\ell), t}] \Big).
\end{align*}

By \cref{lem:option-value} (Option Value Comparison) at period $t+1$:
\begin{align*}
&J^P_{t+1}(\ell^*_t(\ell)) - \mathbb{E}[J^P_{t+1}(\bar{\ell}_t) \mid \mathcal{F}_{\ell^*_t(\ell), t+1}] \\[6pt]
&\qquad \geq \frac{1}{\partial_+g(0)} \Big( J^A_{t+1}(\ell^*_t(\ell)) - \mathbb{E}[J^A_{t+1}(\bar{\ell}_t) \mid \mathcal{F}_{\ell^*_t(\ell), t+1}] \Big).
\end{align*}

Adding these inequalities (with weight $\beta$ on the second):
\begin{align*}
&\Big(V(\mu_{\ell^*_t(\ell), t}, \ell^*_t(\ell)) + \beta J^P_{t+1}(\ell^*_t(\ell))\Big) 
- \mathbb{E}[V(\mu_{\bar{\ell}_t, t}, \bar{\ell}_t) + \beta J^P_{t+1}(\bar{\ell}_t) \mid \mathcal{F}_{\ell^*_t(\ell), t}] \\[6pt]
&\qquad \geq \frac{1}{\partial_+g(0)} \bigg( \Big(U(\mu_{\ell^*_t(\ell), t}, \ell^*_t(\ell)) + \beta J^A_{t+1}(\ell^*_t(\ell))\Big) \\
&\qquad\qquad\qquad\qquad - \mathbb{E}[U(\mu_{\bar{\ell}_t, t}, \bar{\ell}_t) + \beta J^A_{t+1}(\bar{\ell}_t) \mid \mathcal{F}_{\ell^*_t(\ell), t}] \bigg) \\[6pt]
&\qquad \geq 0,
\end{align*}
where the last inequality follows from the agent's optimality condition and the fact that $\partial_+g(0) > 0$.

Thus:
\[
1_{E_t^c}\left(
V(\mu_{\ell^*_t(\ell), t}, \ell^*_t(\ell)) + \beta J^P_{t+1}(\ell^*_t(\ell))
- \mathbb{E}[V(\mu_{\bar{\ell}_t, t}, \bar{\ell}_t) + \beta J^P_{t+1}(\bar{\ell}_t) \mid \mathcal{F}_{\ell^*_t(\ell), t}]
\right)\geq 0.
\]

Taking $\mathbb{E}[\cdot \mid \mathcal{F}_{\ell, t}]$ and using the decomposition of $J^P_t(\ell)$:
\begin{align*}
J^P_t(\ell)
&= \mathbb{E}\!\left[1_{E_t}\!\left(V(\mu_{\bar{\ell}_t,t},\bar{\ell}_t)+\beta J^P_{t+1}(\bar{\ell}_t)\right)
+1_{E_t^c}\!\left(V(\mu_{\ell^*_t(\ell),t},\ell^*_t(\ell))+\beta J^P_{t+1}(\ell^*_t(\ell))\right) \,\Big|\, \mathcal{F}_{\ell,t}\right] \\
&\geq \mathbb{E}[V(\mu_{\bar{\ell}_t, t}, \bar{\ell}_t) + \beta J^P_{t+1}(\bar{\ell}_t) \mid \mathcal{F}_{\ell,t}].
\end{align*}

By the inductive hypothesis applied at $\bar{\ell}_t$:
\[
J^P_{t+1}(\bar{\ell}_t) \geq \mathbb{E}\left[\sum_{s=t+1}^{T} \beta^{s-(t+1)} V(\mu_{\bar{\ell}_s, s}, \bar{\ell}_s) \,\Big|\, \mathcal{F}_{\bar{\ell}_t, t+1}\right].
\]

Substituting:
\begin{align*}
J^P_t(\ell) 
&\geq \mathbb{E}\left[V(\mu_{\bar{\ell}_t, t}, \bar{\ell}_t) + \beta \cdot \mathbb{E}\left[\sum_{s=t+1}^{T} \beta^{s-(t+1)} V(\mu_{\bar{\ell}_s, s}, \bar{\ell}_s) \,\Big|\, \mathcal{F}_{\bar{\ell}_t, t+1}\right] \,\Big|\, \mathcal{F}_{\ell,t}\right] \\[6pt]
&= \mathbb{E}\left[V(\mu_{\bar{\ell}_t, t}, \bar{\ell}_t) + \sum_{s=t+1}^{T} \beta^{s-t} V(\mu_{\bar{\ell}_s, s}, \bar{\ell}_s) \,\Big|\, \mathcal{F}_{\ell,t}\right] \quad \text{(tower property)} \\[6pt]
&= \mathbb{E}\left[\sum_{s=t}^{T} \beta^{s-t} V(\mu_{\bar{\ell}_s, s}, \bar{\ell}_s) \,\Big|\, \mathcal{F}_{\ell,t}\right]
\end{align*}
which completes the inductive step. 
\end{proof}


\medskip
\underline{\textbf{Step 4: Completing the proof of Theorem 1.}}

Applying \cref{lem:period-t} at $t = 1$ from the initial level $0$:
\[
J^P_1(0) \geq \mathbb{E}\left[\sum_{t=1}^{T} \beta^{t-1} V(\mu_{\bar{\ell}_t, t}, \bar{\ell}_t)\right].
\]

The left-hand side is the principal's payoff under the agent's optimal path; the right-hand side is the value of direct control. This direct-control value remains an upper bound even when the regulator can choose any path-adapted mechanism. Indeed, consider the admissible learning process $\mathcal F=\mathcal G$. For any $\phi\in\Phi$, the induced agent path is then $\mathcal G$-adapted. Since transfers affect only the agent's incentives and do not enter the principal's payoff, the principal's payoff from that induced path is no larger than the direct-control value, which maximizes over all feasible $\mathcal G$-adapted paths. The adaptive speed limit $\bm{\phi}^*$ achieves this bound and is therefore learning-robust. Since $\bm{\phi}^*$ does not depend on $U$, it is also dually-robust.
\end{proof}

\clearpage 

\begin{proof}[Proof of \cref{thrm:speed_timecon}]
Throughout the proof, write $h=(\pmb{\ell}_{<t},\ell')$ for a date-$t$ decision history and let $\mathcal G_h:=\mathcal G_{\ell',t}$ and $\mathcal F_h:=\mathcal F_{\ell',t}$. The same convention applies to later decision histories. We also write $h\oplus \pmb q_{\leq s}$ for $(\pmb{\ell}_{<t},\ell')\oplus \pmb q_{\leq s}$. For a continuation mechanism $\phi$ after a date-$t$ history $h$, write the realized discounted continuation transfer along $\pmb q$ as
\[
\Gamma_h^\phi(\pmb q)
:=
\sum_{s=t}^T \beta^{s-t}\phi_s(h\oplus \pmb q_{\leq s}).
\]
Also write
\begin{align*}
\mathcal A_h(\pmb q)
&:=
\sum_{s=t}^T\beta^{s-t}U(\mu_{q_s,s},q_s),\\
\mathcal P_h(\pmb q)
&:=
\sum_{s=t}^T\beta^{s-t}V(\mu_{q_s,s},q_s).
\end{align*}
For a later date-$s$ decision history $h'$, $\Gamma_{h'}^\phi$, $\mathcal A_{h'}$, and $\mathcal P_{h'}$ are defined analogously, with continuation sums starting at $s$.

\underline{\textbf{\smash{Step 1: Conditional characterization.}}}

Since the technology space is finite, conditional direct-control continuation problems have largest solutions after every decision history. Let $
\bar{\pmb q}^{\,h}
=
(\bar q^{\,h}_s)_{s=t}^T$ denote the largest $\mathcal G$-adapted solution to the principal's direct-control continuation problem after $h$. If a later decision history lies on this conditional direct-control path, the continuation solution is the corresponding tail: if $h'$ is reached at date $s$ by $\bar{\pmb q}^{\,h}$, then $\bar q^{\,h'}_r=\bar q^{\,h}_r$ for every $r\geq s$.

\begin{lemma}[Ordered continuation payoff comparison]\label{lem:ordered-continuation-comparison}
Fix a decision history $h$. If $\pmb q^-,\pmb q^+$ are feasible continuations from $h$ with $\pmb q^-\leq\pmb q^+$, then
\[
\mathbb E[\mathcal P_h(\pmb q^-)-\mathcal P_h(\pmb q^+)\mid\mathcal G_h]
\geq
\frac{1}{\partial_+g(0)}
\mathbb E[\mathcal A_h(\pmb q^-)-\mathcal A_h(\pmb q^+)\mid\mathcal G_h].
\]
\end{lemma}

\begin{proof}[Proof of \cref{lem:ordered-continuation-comparison}]
For each date $s\geq t$, apply \cref{lem:magnitude_compstat} to the stopped levels $q^-_s\leq q^+_s$:
\begin{align*}
V(\mu_{q^-_s,s},q^-_s)
&-
\mathbb E[V(\mu_{q^+_s,s},q^+_s)\mid\mathcal F_{q^-_s,s}]
\\
&\geq 
\frac{1}{\partial_+g(0)}
\left(
U(\mu_{q^-_s,s},q^-_s)
-
\mathbb E[U(\mu_{q^+_s,s},q^+_s)\mid\mathcal F_{q^-_s,s}]
\right).
\end{align*}
Because $\mathcal G_h\subseteq\mathcal F_{q^-_s,s}$, taking conditional expectations with respect to $\mathcal G_h$, multiplying by $\beta^{s-t}$, and summing over $s$ gives the result.
\end{proof}

\begin{definition}[Conditional limit-with-option mechanisms]\label{def:conditional-lwo}
Fix a decision history $h$. A continuation mechanism after $h$ is a \emph{conditional limit-with-option mechanism} if: (i) every continuation from $h$ that first exceeds $\bar{\pmb q}^{\,h}$ has infinite penalty; and (ii) at every decision history $h'$ reached from $h$, including $h$ itself, that lies weakly below the conditional limit, the boundary continuation has weakly lowest agent-side expected continuation transfer for every admissible learning process:
\[
\mathbb E[\Gamma_{h'}^\phi(\bar{\pmb q}^{\,h'})\mid\mathcal F_{h'}]
\leq
\mathbb E[\Gamma_{h'}^\phi(\pmb q)\mid\mathcal F_{h'}]
\]
for every feasible below-limit continuation $\pmb q$ from $h'$.
\end{definition}

\begin{lemma}[Conditional robust characterization]\label{lem:tc-conditional-characterization}
Fix a decision history $h$. A continuation mechanism after $h$ is conditionally dually robust if and only if it is a conditional limit-with-option mechanism after $h$. 
\end{lemma}

\begin{proof}[Proof of \cref{lem:tc-conditional-characterization}]
We first prove sufficiency. Fix $(\mathcal F,U)$, and let $\pmb q^*$ be the largest optimal continuation induced by a conditional limit-with-option mechanism $\phi$ after $h$. The hard limit gives $\pmb q^*\leq \bar{\pmb q}^{\,h}$. Since moving to $\bar{\pmb q}^{\,h}$ is feasible from $h$, agent optimality gives the corresponding net-payoff inequality at the agent's information:
\[
\mathbb E[\mathcal A_h(\pmb q^*)-\Gamma_h^\phi(\pmb q^*)\mid\mathcal F_h]
\geq
\mathbb E[\mathcal A_h(\bar{\pmb q}^{\,h})-\Gamma_h^\phi(\bar{\pmb q}^{\,h})\mid\mathcal F_h].
\]
The option condition gives
\[
\mathbb E[\Gamma_h^\phi(\bar{\pmb q}^{\,h})\mid\mathcal F_h]
\leq
\mathbb E[\Gamma_h^\phi(\pmb q^*)\mid\mathcal F_h],
\]
and hence
\[
\mathbb E[\mathcal A_h(\pmb q^*)\mid\mathcal F_h]
\geq
\mathbb E[\mathcal A_h(\bar{\pmb q}^{\,h})\mid\mathcal F_h].
\]
Taking conditional expectations down to $\mathcal G_h$ yields
\[
\mathbb E[\mathcal A_h(\pmb q^*)\mid\mathcal G_h]
\geq
\mathbb E[\mathcal A_h(\bar{\pmb q}^{\,h})\mid\mathcal G_h].
\]
Since $\pmb q^*\leq\bar{\pmb q}^{\,h}$, \cref{lem:ordered-continuation-comparison} then implies
\[
\mathbb E[\mathcal P_h(\pmb q^*)-\mathcal P_h(\bar{\pmb q}^{\,h})\mid\mathcal G_h]
\geq
\frac{1}{\partial_+g(0)}
\mathbb E[\mathcal A_h(\pmb q^*)-\mathcal A_h(\bar{\pmb q}^{\,h})\mid\mathcal G_h]
\geq 0.
\]
Thus every admissible $(\mathcal F,U)$ yields at least the conditional direct-control payoff. Under the admissible learning process $\mathcal F=\mathcal G$, no continuation mechanism can exceed the conditional direct-control payoff because the induced continuation is $\mathcal G$-adapted and transfers do not enter the principal's payoff. Therefore $\phi$ is conditionally dually robust.

We now prove necessity. If either condition fails, let $h'$ be the first history at which it fails. Conditional on $\mathcal G_{h'}$, the continuation problem from $h'$ has the same order, pasting, and payoff-comparison structure as the one-dimensional stopped-level problem in \citet{koh2024robust}.

First, suppose the hard limit fails: with positive probability, some feasible continuation first exceeds $\bar{\pmb q}^{\,h}$ while paying finite transfer. Let $h'$ be the first-crossing decision history. Conditional on $\mathcal G_{h'}$, the continuation problem is the one-dimensional stopped-level problem considered in \citet{koh2024robust}, with the Snell continuation value replacing the one-period payoff. The construction in \citet{koh2024robust} chooses an admissible learning process and a sufficiently risk-seeking preference so that the agent strictly prefers the finite-transfer continuation beyond the conditional direct-control boundary. Because $\bar{\pmb q}^{\,h'}$ is the largest conditional direct-control solution, this lowers the principal's conditional payoff below the direct-control guarantee on a positive-probability event, contradicting conditional dual robustness.

Second, suppose the option condition fails at some below-limit decision history $h'$: for some admissible learning process and some feasible below-limit continuation $\pmb q$,
\[
\mathbb E[\Gamma_{h'}^\phi(\pmb q)\mid\mathcal F_{h'}]
<
\mathbb E[\Gamma_{h'}^\phi(\bar{\pmb q}^{\,h'})\mid\mathcal F_{h'}]
\]
on a positive-probability event in $\mathcal F_{h'}$. Conditional on $\mathcal G_{h'}$, apply the rare-good-news construction from \citet{koh2024robust} to the continuation problem: after $h'$, the agent learns a small-probability event on which the state is good and the agent-side expected-transfer wedge favoring $\pmb q$ over the boundary continuation is revealed. On a positive-probability subevent, let the wedge be bounded below by $\delta>0$. Since the continuation problem is finite and payoffs are bounded, there is $M<\infty$ such that the unscaled payoff difference between any two feasible continuations is at most $M$. 

Choose the linear preference $U_\eta(\mu,l)=\eta V(\mu,l)$ with $0<\eta<\delta/(M\vee 1)$. Then the gross payoff gain from moving from $\pmb q$ to $\bar{\pmb q}^{\,h'}$ is smaller than $\delta$, so the agent-side expected-transfer wedge makes the agent choose the lower continuation $\pmb q$ rather than $\bar{\pmb q}^{\,h'}$. The aligned-sensitivity comparison then makes the principal strictly worse than under the boundary continuation on that event, again contradicting the conditional direct-control guarantee. Hence both the hard limit and the option condition are necessary.
\end{proof}

\underline{\textbf{\smash{Step 2: Local monotone selection under flattening.}}}
We begin with some notation. For a decision history $h$, let
\[
\mathcal Q(h)
:=
\left\{
\pmb q:
\pmb q \text{ is feasible from } h
\text{ and }
\pmb q\leq \bar{\pmb q}^{\,h}
\right\}
\]
be the below-limit continuation set.

Whenever a below-limit decision history $h$ is fixed in the local argument, use the following notation. Write $h=(\pmb{\ell}_{<t},l)$. For each current-period choice $s\in\mathcal L$ with
\[
l\leq s\leq \bar q_t^{\,h},
\]
let $\operatorname{loc}_h(s)$ denote the local alternative generated by choosing $q_t=s$ at date $t$. If $t<T$, this is the date-$(t+1)$ decision history with realized past $(\pmb{\ell}_{<t},s)$ and current level $s$. If $t=T$, $\operatorname{loc}_h(s)$ is the terminal history generated by the time-$T$ choice $q_T=s$. Let
\[
\mathcal S(h)
:=
\{\operatorname{loc}_h(s):s\in\mathcal L,\ l\leq s\leq \bar q_t^{\,h}\}.
\]
This is a finite chain, ordered by the current-period choice $s$. Waiting is the special case $s=l$, and scaling corresponds to $s>l$. For $\operatorname{loc}_h(s)\in\mathcal S(h)$, let
\[
\mathcal Q(h;\operatorname{loc}_h(s))
:=
\{\pmb q\in\mathcal Q(h):q_t=s\}.
\]
Thus the sets $\mathcal Q(h;\operatorname{loc}_h(s))$ partition $\mathcal Q(h)$.

For this fixed $h$, say that expected transfers are \emph{locally decreasing} if, for every admissible learning process, every $h^-,h^+\in\mathcal S(h)$ with $h^-\leq h^+$, every $\pmb q^-\in\mathcal Q(h;h^-)$, and every $\pmb q^+\in\mathcal Q(h;h^+)$,
\[
\mathbb E[\Gamma_h^\phi(\pmb q^+)\mid\mathcal F_h]
\leq
\mathbb E[\Gamma_h^\phi(\pmb q^-)\mid\mathcal F_h].
\tag{LD}\label{eq:tc-local-decreasing}
\]

\begin{lemma}[Selection on the local chain]\label{lem:tc-local-selection}
Fix a below-limit decision history $h$ and suppose \eqref{eq:tc-local-decreasing} holds. Let $\phi'$ be a continuation mechanism with the same conditional hard limit as $\phi$ such that, for every admissible learning process, expected continuation transfers under $\phi'$ are flat across the local alternatives:
\[
\mathbb E[\Gamma_h^{\phi'}(\pmb q^+)\mid\mathcal F_h]
=
\mathbb E[\Gamma_h^{\phi'}(\pmb q^-)\mid\mathcal F_h]
\tag{LF}\label{eq:tc-local-flat}
\]
for all $h^+,h^-\in\mathcal S(h)$, all
$\pmb q^+\in\mathcal Q(h;h^+)$, and all
$\pmb q^-\in\mathcal Q(h;h^-)$.
Fix $(\mathcal F,U)$. Let $\pmb q^\phi$ and $\pmb q^{\phi'}$ be the selected largest optimal continuations under $\phi$ and $\phi'$, and let $h_\phi^+$ and $h_{\phi'}^+$ be the local alternatives in $\mathcal S(h)$ induced by their current-period choices. Then $h_{\phi'}^+\leq h_\phi^+$ and $\mathbb E[\mathcal A_h(\pmb q^{\phi'})\mid\mathcal F_h]
\geq
\mathbb E[\mathcal A_h(\pmb q^\phi)\mid\mathcal F_h]$. 
\end{lemma}

\begin{proof}[Proof of \cref{lem:tc-local-selection}]
Since $\phi'$ has flat expected transfers on $\mathcal Q(h)$, optimality of $\pmb q^{\phi'}$ under $\phi'$ implies the gross-payoff inequality
\[
\mathbb E[\mathcal A_h(\pmb q^{\phi'})\mid\mathcal F_h]
\geq
\mathbb E[\mathcal A_h(\pmb q^\phi)\mid\mathcal F_h].
\]
Suppose toward a contradiction that $h_{\phi'}^+>h_\phi^+$. Local decreasing under $\phi$ gives
\[
\mathbb E[\Gamma_h^\phi(\pmb q^{\phi'})\mid\mathcal F_h]
\leq
\mathbb E[\Gamma_h^\phi(\pmb q^\phi)\mid\mathcal F_h].
\]
Combining the two displays,
\[
\mathbb E[\mathcal A_h(\pmb q^{\phi'})-\Gamma_h^\phi(\pmb q^{\phi'})\mid\mathcal F_h]
\geq
\mathbb E[\mathcal A_h(\pmb q^\phi)-\Gamma_h^\phi(\pmb q^\phi)\mid\mathcal F_h].
\]
Hence $\pmb q^{\phi'}$ is also optimal under $\phi$ and reaches a strictly higher local alternative than the selected largest $\phi$-optimal continuation. This contradicts the selection of $\pmb q^\phi$. Therefore $h_{\phi'}^+\leq h_\phi^+$, and the gross-payoff inequality was already shown.
\end{proof}

\begin{lemma}[Local option-value comparison]\label{lem:tc-local-option}
Fix a below-limit decision history $h=(\pmb{\ell}_{<t},l)$ and two local alternatives $h^-,h^+\in\mathcal S(h)$ with $h^-\leq h^+$. Suppose that for every below-limit decision history $h'\succ h$, 
\[
\mathbb E[\Gamma_{h'}^\phi(\pmb r)\mid\mathcal F_{h'}]
=
\mathbb E[\Gamma_{h'}^\phi(\bar{\pmb q}^{\,h'})\mid\mathcal F_{h'}]
\qquad
\forall \pmb r\in\mathcal Q(h').
\]
Let $\pmb q^-\in\mathcal Q(h;h^-)$ and $\pmb q^+\in\mathcal Q(h;h^+)$ be continuation technology paths whose tails after the local choices maximize the agent's payoff gross of transfers.\footnote{That is, conditional on $h^\pm$, the continuation tail of $\pmb q^\pm$ is a largest maximizer of $\mathcal A_{h^\pm}$ over $\mathcal Q(h^\pm)$; if $t=T$, this condition is vacuous.} Then
\[
\mathbb E[\mathcal P_h(\pmb q^-)-\mathcal P_h(\pmb q^+)\mid\mathcal G_h]
\geq
\frac{1}{\partial_+g(0)}
\mathbb E[\mathcal A_h(\pmb q^-)-\mathcal A_h(\pmb q^+)\mid\mathcal G_h].
\]
\end{lemma}

\begin{proof}[Proof of \cref{lem:tc-local-option}]
Write $h^-=\operatorname{loc}_h(s^-)$ and $h^+=\operatorname{loc}_h(s^+)$, with $s^-\leq s^+$. The current-period choices are therefore $q_t^-=s^-$ and $q_t^+=s^+$. If $t=T$, there is no continuation tail, and the claim is exactly \cref{lem:magnitude_compstat}, conditioned down to $\mathcal G_h$.

Suppose $t<T$, and write $\pmb q^\pm_{>t}$ for the continuation tails of $\pmb q^\pm$ after the local choices. For any strict successor history $h'\succ h$, define the gross-agent continuation value
\[
J_A(h')
:=
\operatorname*{ess\,sup}_{\pmb r\in\mathcal Q(h')}
\mathbb E[\mathcal A_{h'}(\pmb r)\mid\mathcal F_{h'}].
\]
Let $\pmb r^*(h')$ be the largest maximizer of this problem, and define the associated principal continuation value
\[
J_P(h')
:=
\mathbb E[\mathcal P_{h'}(\pmb r^*(h'))\mid\mathcal F_{h'}].
\]
The largest maximizer exists by the same compactness and stopped-level pasting argument used above. The displayed successor-flatness assumption gives, for each such $h'$, an $\mathcal F_{h'}$-measurable constant
\[
C(h')
:=
\mathbb E[\Gamma_{h'}^\phi(\bar{\pmb q}^{\,h'})\mid\mathcal F_{h'}]
\]
such that $\mathbb E[\Gamma_{h'}^\phi(\pmb r)\mid\mathcal F_{h'}]=C(h')$ for every $\pmb r\in\mathcal Q(h')$. Thus, after a strict successor history, maximizing the agent's expected payoff net of transfers is the same as maximizing the agent's expected payoff gross of transfers. Since $\pmb q^-$ and $\pmb q^+$ maximize the agent's gross payoff,
\begin{align}
\mathbb E[\mathcal A_{h^\pm}(\pmb q^\pm_{>t})\mid\mathcal F_{h^\pm}]
&=J_A(h^\pm),
\label{eq:local-gross-tail-A}\\
\mathbb E[\mathcal P_{h^\pm}(\pmb q^\pm_{>t})\mid\mathcal F_{h^\pm}]
&=J_P(h^\pm).
\label{eq:local-gross-tail-P}
\end{align}

We next compare the current-period flow payoffs. Applying \cref{lem:magnitude_compstat} to the date-$t$ stopped levels $s^-\leq s^+$ and then conditioning on $\mathcal G_h$ gives
\begin{equation}\label{eq:local-flow-comparison}
\mathbb E[
V(\mu_{s^-,t},s^-)-V(\mu_{s^+,t},s^+)
\mid\mathcal G_h]
\geq
\frac{1}{\partial_+g(0)}
\mathbb E[
U(\mu_{s^-,t},s^-)-U(\mu_{s^+,t},s^+)
\mid\mathcal G_h].
\end{equation}
For the tails, apply the option-value comparison in \cref{lem:option-value} from date $t+1$, holding the past history $\pmb{\ell}_{<t}$ fixed and using $s^-\leq s^+$ as the two starting levels. This gives 
\[
J_P(h^-)-\mathbb E[J_P(h^+)\mid\mathcal F_{h^-}]
\geq
\frac{1}{\partial_+g(0)}
\left(
J_A(h^-)-\mathbb E[J_A(h^+)\mid\mathcal F_{h^-}]
\right).
\]
Conditioning on $\mathcal G_h$ yields
\begin{equation}\label{eq:local-tail-comparison}
\mathbb E[J_P(h^-)-J_P(h^+)\mid\mathcal G_h]
\geq
\frac{1}{\partial_+g(0)}
\mathbb E[J_A(h^-)-J_A(h^+)\mid\mathcal G_h].
\end{equation}

Finally, decompose the continuation payoffs into the current-period flow and the tail:
\begin{align*}
\mathcal A_h(\pmb q^\pm)
&=
U(\mu_{s^\pm,t},s^\pm)
+\beta\,\mathcal A_{h^\pm}(\pmb q^\pm_{>t}),\\
\mathcal P_h(\pmb q^\pm)
&=
V(\mu_{s^\pm,t},s^\pm)
+\beta\,\mathcal P_{h^\pm}(\pmb q^\pm_{>t}).
\end{align*}
Using \eqref{eq:local-gross-tail-A}, \eqref{eq:local-gross-tail-P}, and the tower property, the conditional differences in $\mathcal A_h$ and $\mathcal P_h$ are the corresponding flow differences plus $\beta$ times the differences in $J_A$ and $J_P$. Adding \eqref{eq:local-flow-comparison} and $\beta$ times \eqref{eq:local-tail-comparison} gives the desired inequality.
\end{proof}

\underline{\textbf{\smash{Step 3: Dominance of local downward transfers.}}}

\begin{lemma}[Dominance of downward continuation transfers]\label{lem:tc-downward-dominance}
Fix a below-limit decision history $h=(\pmb{\ell}_{<t},l)$. Let $\phi$ be a continuation mechanism with the conditional hard limit $\bar{\pmb q}^{\,h}$ and finite realized continuation transfers on $\mathcal Q(h)$. Suppose that for every below-limit decision history $h'\succ h$, where $h'$ has current date-level pair $(l',t')$ and $h'\succ h$ means $t'>t$ or $(t'=t$ and $l'>l)$,
\[
\mathbb E[\Gamma_{h'}^\phi(\pmb r)\mid\mathcal F_{h'}]
=
\mathbb E[\Gamma_{h'}^\phi(\bar{\pmb q}^{\,h'})\mid\mathcal F_{h'}]
\qquad
\forall \pmb r\in\mathcal Q(h').
\]
Suppose also that expected transfers are locally decreasing at $h$ in the sense of \eqref{eq:tc-local-decreasing}. Let $\phi'$ be a local flattening after $h$: it preserves the same hard limit and satisfies \eqref{eq:tc-local-flat}. Then $\phi'$ weakly dominates $\phi$ after $h$. If the local decreasing inequality is strict for some pair of local alternatives on a positive-probability event, then the dominance is strict for some admissible $(\mathcal F,U)$.
\end{lemma}

\begin{proof}[Proof of \cref{lem:tc-downward-dominance}]
First note that such a flattening can be implemented by an admissible path-adapted continuation mechanism. To keep public measurability explicit, one may use a terminal exact implementation: preserve the same $+\infty$ transfer for paths that exceed the conditional hard limit, and for each $\pmb q\in\mathcal Q(h)$ set the terminal flow transfer so that
\[
\Gamma_h^{\phi'}(\pmb q)
=
\mathbb E[\Gamma_h^\phi(\bar{\pmb q}^{\,h})\mid\mathcal G_h].
\]
When $\beta>0$, this is done by changing $\phi_T(h\oplus\pmb q_{\leq T})$ by the discounted difference; when $\beta=0$, the analogous adjustment is made at date $t$. The adjustment is measurable at the realized truncation because the mechanism observes the whole terminal below-limit path, and $\mathcal G_h\subseteq\mathcal G_{q_T,T}$. Since $\mathcal L$ is finite, lower semicontinuity is automatic on the finite-valued region and integrability is inherited from $\phi$. This implementation is stronger than the local flatness property used below, because it makes realized total transfers flat on $\mathcal Q(h)$.

Fix $(\mathcal F,U)$. Let $\pmb q^\phi$ and $\pmb q^{\phi'}$ be the selected largest optimal continuations after $h$ under $\phi$ and $\phi'$, and let $h_\phi^+$ and $h_{\phi'}^+$ be the local alternatives induced by their current-period choices. By \cref{lem:tc-local-selection},
\[
h_{\phi'}^+\leq h_\phi^+
\qquad\text{and}\qquad
\mathbb E[\mathcal A_h(\pmb q^{\phi'})\mid\mathcal F_h]
\geq
\mathbb E[\mathcal A_h(\pmb q^\phi)\mid\mathcal F_h].
\]
Conditioning the gross-payoff inequality down to $\mathcal G_h$ gives
\[
\mathbb E[\mathcal A_h(\pmb q^{\phi'})-\mathcal A_h(\pmb q^\phi)\mid\mathcal G_h]
\geq
0.
\]
By the displayed successor-flatness hypothesis of this lemma, transfers under $\phi$ are constant after every strict successor of $h$; by construction, the flattening makes total transfers under $\phi'$ constant on $\mathcal Q(h)$. Hence the selected tails after $h_{\phi'}^+$ and $h_\phi^+$ are gross-agent-payoff maximizing tails. Applying \cref{lem:tc-local-option} with $h^-=h_{\phi'}^+$, $h^+=h_\phi^+$, $\pmb q^-=\pmb q^{\phi'}$, and $\pmb q^+=\pmb q^\phi$ yields
\[
\mathbb E[\mathcal P_h(\pmb q^{\phi'})-\mathcal P_h(\pmb q^\phi)\mid\mathcal G_h]
\geq
\frac{1}{\partial_+g(0)}
\mathbb E[\mathcal A_h(\pmb q^{\phi'})-\mathcal A_h(\pmb q^\phi)\mid\mathcal G_h]
\geq 0.
\]
Thus $\phi'$ weakly dominates $\phi$ after $h$.

For strictness, take a positive-probability event on which the local decreasing inequality is strict for some $h^-<h^+$ and some $\pmb q^-\in\mathcal Q(h;h^-)$, $\pmb q^+\in\mathcal Q(h;h^+)$. As in \citet{koh2024robust}, choose an admissible bad-news learning process: before the relevant local choice the agent and principal have the same information, and at that choice bad news arrives with strictly positive conditional probability. On the bad-news event, choose a sufficiently small aligned preference scale $U_\eta(\mu,l)=\eta V(\mu,l)$ so that gross continuation-payoff differences are dominated by the strict local transfer wedge. Then under $\phi$ the agent is induced to choose the higher, cheaper local alternative $h^+$, while under the flattened transfer she chooses the lower local alternative $h^-$. The local option-value comparison is strict on the constructed event, so $\phi'$ strictly dominates $\phi$ for that admissible $(\mathcal F,U)$.
\end{proof}

\underline{\textbf{\smash{Step 4: Backward flatness.}}}

\begin{lemma}[Backward flatness]\label{lem:tc-flatness}
Let $\phi$ be dually robust and time-consistent. Then after every decision history $h$ and every admissible learning process, the agent-side expected continuation transfer is flat on the whole below-limit region:
\[
\mathbb E[\Gamma_{h'}^\phi(\pmb q)\mid\mathcal F_{h'}]
=
\mathbb E[\Gamma_{h'}^\phi(\bar{\pmb q}^{\,h'})\mid\mathcal F_{h'}]
\]
for every later below-limit decision history $h'$ reached from $h$ and every $\pmb q\in\mathcal Q(h')$.
Consequently, conditioning down gives the analogous equality conditional on $\mathcal G_{h'}$.
\end{lemma}

\begin{proof}[Proof of \cref{lem:tc-flatness}]
First note why \cref{lem:tc-conditional-characterization} applies at continuation histories. If the continuation after any decision history failed the conditional characterization, the necessity argument in \cref{lem:tc-conditional-characterization}, pasted after that history, would produce an admissible conditional environment in which the continuation falls below the conditional direct-control guarantee. Replacing that continuation by the conditional speed-limit continuation would then conditionally dominate it, contradicting time-consistency. Hence, after every decision history, $\phi$ has the conditional hard limit and satisfies the weak option to the limit. Since $\mathcal L$ is finite, write its levels as $l_0<l_1<\cdots<l_N$. We prove flatness for all below-limit decision histories by lexicographic backward induction on the finite time-level grid: later dates precede earlier dates, and within each date higher levels precede lower levels.
\begin{center}
\begin{tikzpicture}[
    x=1.15cm,
    y=.65cm,
    cell/.style={circle, draw=black!55, fill=white, minimum size=4pt, inner sep=0pt},
    induction/.style={-{Stealth[length=2mm]}, thick, blue!65!black},
    faint/.style={black!18}
]
\draw[->, black!60] (-.45,-.45) -- (3.55,-.45) node[right] {\small date};
\draw[->, black!60] (-.45,-.45) -- (-.45,3.45) node[above] {\small level};
\node[below] at (0,-.45) {$1$};
\node[below] at (1,-.45) {$2$};
\node[below] at (2,-.45) {$\cdots$};
\node[below] at (3,-.45) {$T$};
\node[left] at (-.45,0) {$l_0$};
\node[left] at (-.45,1) {$\vdots$};
\node[left] at (-.45,2) {$l_{N-1}$};
\node[left] at (-.45,3) {$l_N$};
\foreach \x in {0,1,3} {
    \draw[faint] (\x,0) -- (\x,3);
}
\draw[faint, densely dotted] (2,0) -- (2,3);
\foreach \y in {0,2,3} {
    \draw[faint] (0,\y) -- (3,\y);
}
\draw[faint, densely dotted] (0,1) -- (3,1);
\foreach \x in {0,1,3} {
    \foreach \y in {0,2,3} {
        \node[cell] at (\x,\y) {};
    }
    \node at (\x,1) {$\vdots$};
}
\foreach \y in {0,2,3} {
    \node at (2,\y) {$\cdots$};
}
\node[above right, blue!65!black] at (3,3) {\scriptsize start};
\draw[induction] (3,2.85) -- (3,2.15);
\draw[induction] (3,1.85) -- (3,.15);
\draw[induction] (2.95,-.05) to[out=-130,in=50] (1.05,3.05);
\draw[induction] (1,2.85) -- (1,2.15);
\node[blue!65!black, align=center, font=\scriptsize] at (1.55,3.55) {later dates first;\\higher levels first};
\end{tikzpicture}
\end{center}
For time-level cells, write
\[
(l',t')\succ(l,t)
\quad\Longleftrightarrow\quad
t'>t
\text{ or }
(t'=t\text{ and }l'>l).
\]
The grid successors of $(l,t)$ are simply the cells that follow it in this order:
\[
\{(l',t'):(l',t')\succ(l,t)\}.
\]
Equivalently, they are either cells at a later date $t'>t$, or cells at the same date $t'=t$ with a higher level $l'>l$. A decision history has successor cell $(l',t')$ if it extends the current history and its current date-level pair is $(l',t')$.

\emph{Base step.} Consider the cell $(l_N,T)$. If no below-limit decision history has this cell, the claim is vacuous. Otherwise, let $h$ be such a history. The current level is the largest feasible level and there is no later date. Hence $\mathcal Q(h)$ contains only the singleton continuation, which is also $\bar{\pmb q}^{\,h}$, and flatness is immediate.

\emph{Inductive step.} Fix a cell $(l,t)$ and suppose the following inductive hypothesis is true:
for every grid successor $(l',t')\succ(l,t)$, every below-limit decision history $h'$ whose current date-level pair is $(l',t')$, and every $\pmb q\in\mathcal Q(h')$,
\[
\mathbb E[\Gamma_{h'}^\phi(\pmb q)\mid\mathcal F_{h'}]
=
\mathbb E[\Gamma_{h'}^\phi(\bar{\pmb q}^{\,h'})\mid\mathcal F_{h'}].
\]
Let $h=(\pmb{\ell}_{<t},l)$ be any below-limit decision history whose current date-level pair is $(l,t)$, and recall the notation $\mathcal S(h)$ and $\mathcal Q(h;\operatorname{loc}_h(s))$ from Step 2. The induction is over the whole time-level grid of decision histories, so same-date higher histories such as $(\pmb{\ell}_{<t},s)$ with $s>l$ are included among the strict grid successors of $h$.

If $t<T$, the induction hypothesis covers every date-$(t+1)$ history $\operatorname{loc}_h(s)\in\mathcal S(h)$. If $t=T$, there is no downstream tail. Once the current-period choice $q_t=s$ is fixed, the current transfer is fixed within the class $\mathcal Q(h;\operatorname{loc}_h(s))$. The induction hypothesis then implies that the remaining expected transfer is independent of the downstream tail the agent chooses. Thus, within each class $\mathcal Q(h;\operatorname{loc}_h(s))$, all continuations have the same expected transfer from $h$.

We next verify the local decreasing condition \eqref{eq:tc-local-decreasing}. Fix $l\leq s^-\leq s^+\leq\bar q_t^{\,h}$, write $h^-=\operatorname{loc}_h(s^-)$ and $h^+=\operatorname{loc}_h(s^+)$, and take $\pmb q^-\in\mathcal Q(h;h^-)$ and $\pmb q^+\in\mathcal Q(h;h^+)$. If $s^-=s^+$, the desired comparison follows from the within-class equality just proved. Suppose first that both choices involve positive scaling, so $l<s^-<s^+$. Let $\tilde h_{s^-}=(\pmb{\ell}_{<t},s^-)$ be the same-date decision history at level $s^-$. Its cell $(s^-,t)$ is a strict successor of $(l,t)$, so the induction hypothesis applies at $\tilde h_{s^-}$. From $\tilde h_{s^-}$, compare the continuation that chooses $s^-$ at date $t$ and then follows the tail of $\pmb q^-$ with the continuation that chooses $s^+$ at date $t$ and then follows the tail of $\pmb q^+$. These two continuations are below the conditional limit and generate the same realized path truncations as $\pmb q^-$ and $\pmb q^+$ do from $h$. Flatness at $\tilde h_{s^-}$ therefore gives equality of their expected transfers conditional on $\mathcal F_{\tilde h_{s^-}}$; taking conditional expectations down to $\mathcal F_h$ gives
\[
\mathbb E[\Gamma_h^\phi(\pmb q^+)\mid\mathcal F_h]
=
\mathbb E[\Gamma_h^\phi(\pmb q^-)\mid\mathcal F_h].
\]

It remains to compare waiting, $s^-=l$, with a positive scale-up choice $s^+>l$. Let $b=\bar q_t^{\,h}$ be the boundary current-period choice. The weak option condition at $h$ gives
\[
\mathbb E[\Gamma_h^\phi(\bar{\pmb q}^{\,h})\mid\mathcal F_h]
\leq
\mathbb E[\Gamma_h^\phi(\pmb q^-)\mid\mathcal F_h].
\]
The boundary continuation belongs to $\mathcal Q(h;\operatorname{loc}_h(b))$. If $s^+<b$, the positive-scale equality just proved equates the class $\operatorname{loc}_h(s^+)$ with the boundary class $\operatorname{loc}_h(b)$; if $s^+=b$, the within-class equality gives the same conclusion. Thus every continuation in $\mathcal Q(h;\operatorname{loc}_h(s^+))$ has the same expected transfer as the boundary continuation. Hence
\[
\mathbb E[\Gamma_h^\phi(\pmb q^+)\mid\mathcal F_h]
\leq
\mathbb E[\Gamma_h^\phi(\pmb q^-)\mid\mathcal F_h].
\]
This proves \eqref{eq:tc-local-decreasing} for all $h^-\leq h^+$ in $\mathcal S(h)$.

On $\mathcal Q(h)$, realized continuation transfers are finite by the conditional hard-limit characterization. If this local inequality is strict for some pair of local alternatives on a positive-probability event, then \cref{lem:tc-downward-dominance} gives a continuation mechanism $\phi'$ that strictly dominates $\phi$ after $h$, contradicting time-consistency. Hence no strict local wedge is possible. Expected transfers are therefore flat across all local alternatives in $\mathcal S(h)$:
\[
\mathbb E[\Gamma_h^\phi(\pmb q^-)\mid\mathcal F_h]
=
\mathbb E[\Gamma_h^\phi(\pmb q^+)\mid\mathcal F_h]
\]
for all $h^-,h^+\in\mathcal S(h)$ with $h^-\leq h^+$, all
$\pmb q^-\in\mathcal Q(h;h^-)$, and all
$\pmb q^+\in\mathcal Q(h;h^+)$.
Finally, the boundary continuation $\bar{\pmb q}^{\,h}$ belongs to the class $\mathcal Q(h;\operatorname{loc}_h(b))$, and successor flatness has already made transfers constant within every class. Thus, for every $\pmb q\in\mathcal Q(h)$,
\[
\mathbb E[\Gamma_{h}^\phi(\pmb q)\mid\mathcal F_h]
=
\mathbb E[\Gamma_{h}^\phi(\bar{\pmb q}^{\,h})\mid\mathcal F_h].
\]
This proves the claim at cell $(l,t)$ and completes the induction. Since every below-limit decision history belongs to some cell of the grid, the lemma follows.
\end{proof}

\underline{\textbf{\smash{Step 5: Time-consistency and uniqueness.}}}

\begin{lemma}[Conditional tail of the adaptive speed limit]\label{lem:tc-speed-tail}
Let $h$ be any decision history. The continuation of $\bm{\phi}^*$ after $h$ imposes the conditional direct-control hard limit $\bar{\pmb q}^{\,h}$ and charges zero transfer below that limit.
\end{lemma}

\begin{proof}[Proof of \cref{lem:tc-speed-tail}]
Since the technology space is finite, the direct-control problem can be solved by backward induction over decision histories. We take $\bar{\pmb q}^{\,h}$ to be the largest maximizer in the Bellman problem conditional on $\mathcal G_h$. This Bellman selection is dynamically consistent: if a later decision history $h'$ lies on $\bar{\pmb q}^{\,h}$, then the tail of $\bar{\pmb q}^{\,h}$ from $h'$ is $\bar{\pmb q}^{\,h'}$.

The adaptive speed limit $\bm{\phi}^*$ is generated by this Bellman-consistent boundary at every decision history. Thus, after any decision history $h$, including off-path histories and histories at which an earlier boundary has been crossed, past transfers are sunk and the continuation permits exactly the continuations $\pmb q$ satisfying $\pmb q\leq\bar{\pmb q}^{\,h}$ and assigns $+\infty$ to the first future truncation that exceeds $\bar{\pmb q}^{\,h}$. Below the boundary, every future flow transfer of $\bm{\phi}^*$ is zero. This proves the claim.
\end{proof}

We now prove the theorem. First fix any decision history $h$. To verify time-consistency, let $\phi'$ be an arbitrary candidate deviating continuation mechanism after $h$. Under the admissible learning process $\mathcal F=\mathcal G$, the largest optimal continuation induced by $\phi'$ is $\mathcal G$-adapted, and transfers do not enter the principal's payoff. Hence its conditional principal payoff is bounded above by
\[
\sup_{\pmb q}
\mathbb E\!\left[
\sum_{s=t}^T \beta^{s-t}V(\mu_{q_s,s},q_s)
\;\middle|\;
\mathcal G_h
\right],
\]
where the supremum is over feasible $\mathcal G$-adapted continuations from $h$. The continuation direct-control path $\bar{\pmb q}^{\,h}$ attains this bound:
\[
\sup_{\pmb q}
\mathbb E\!\left[
\sum_{s=t}^T \beta^{s-t}V(\mu_{q_s,s},q_s)
\;\middle|\;
\mathcal G_h
\right]
=
\mathbb E\!\left[
\sum_{s=t}^T \beta^{s-t}V(\mu_{\bar q^{\,h}_s,s},\bar q^{\,h}_s)
\;\middle|\;
\mathcal G_h
\right].
\tag{TC}\label{eq:tc-direct-control-bound}
\]
By \cref{lem:tc-speed-tail}, the continuation of $\bm{\phi}^*$ after $h$ imposes the hard limit $\bar{\pmb q}^{\,h}$ and charges zero transfer below it. This is a statement about the boundary imposed by the mechanism, not about the agent's selected continuation path. The continuation is therefore a conditional limit-with-option mechanism, so \cref{lem:tc-conditional-characterization} implies that it is conditionally dually robust. In particular, under the admissible learning process $\mathcal F=\mathcal G$, it attains the conditional direct-control value in \eqref{eq:tc-direct-control-bound}. Thus no continuation mechanism can give weakly higher conditional principal payoff for every $(\mathcal F,U)$ and a strict improvement for some $(\mathcal F,U)$ after $h$. Since $h$ was arbitrary, $\bm{\phi}^*$ is time-consistent.

For uniqueness, let $\phi$ be any dually robust and time-consistent mechanism. By \cref{lem:tc-conditional-characterization}, after every decision history $h$ its hard limit is the conditional direct-control boundary $\bar{\pmb q}^{\,h}$, and
\[
\Gamma_h^\phi(\pmb q)=+\infty
\qquad
\text{whenever }\pmb q \not\leq \bar{\pmb q}^{\,h}.
\]
By \cref{lem:tc-flatness}, for every admissible learning process there exists an $\mathcal F_h$-measurable random variable $C^{\mathcal F}(h)$ such that
\[
\mathbb E[\Gamma_h^\phi(\pmb q)\mid\mathcal F_h]
=
C^{\mathcal F}(h)
\qquad
\forall \pmb q\leq \bar{\pmb q}^{\,h}.
\]
Conditioning down to $\mathcal G_h$ gives the corresponding constant 
\[
\mathbb E[\Gamma_h^\phi(\bar{\pmb q}^{\,h})\mid\mathcal G_h].
\]
Thus every continuation of $\phi$ is equivalent, for all agent incentive constraints and principal payoff comparisons, to the conditional speed limit plus an additive continuation constant. The constant does not affect the agent's continuation choices or the principal's payoff. After normalizing these payoff-irrelevant continuation constants to zero after every decision history, $\phi$ coincides with $\bm{\phi}^*$. Hence the adaptive speed limit is the unique time-consistent and dually robust mechanism, up to payoff-irrelevant continuation constants. \end{proof}

\end{document}